\documentclass[11pt]{article}

\usepackage{amsmath,amssymb,latexsym}
\usepackage{graphicx}
\usepackage[T1]{fontenc}
\usepackage{pstricks}
\usepackage{pst-node}
\usepackage{url}

\usepackage[version=3]{mhchem} 

\newtheorem{example}{Example}

\newcommand{\lsbrace}{\ensuremath{|\![}}
\newcommand{\rsbrace}{\ensuremath{]\!|}}
\newcommand{\eqdef}{\stackrel{\mbox{\scriptsize \rm def}}{=}}

\newcommand{\comp}{\sim}

\newcommand{\cname}{{\cal A}}

\newcommand{\Act}{{\sf Act}}

\newcommand{\shapes}{\mathbb{S}}
\newcommand{\bonds}{{\sf bonds}}
\newcommand{\Split}{{\sf split}}
\newcommand{\reference}{{\sf global}}

\newcommand{\real}{\mathbb{R}}
\newcommand{\vect}[1]{\mathbf{#1}}
\newcommand{\gor}{\;\big|\;}

\newcommand{\proc}{\sf{3DP}}

\newcommand{\ch}[2]{\ensuremath{\langle #1, #2 \rangle}}
\newcommand{\channels}{{\cal C}}
\newcommand{\pconst}{{\cal K}}

\newcommand{\colliding}{{\sf colliding}}

\newcommand{\ftoc}{{\sf Ftoc}}

\newcommand{\mts}{\Delta}
\newcommand{\bpa}{\mathbb{B}}
\newcommand{\nil}{{\sf nil}}
\newcommand{\NIL}{{\sf Nil}}

\newcommand{\move}{\ensuremath{\mathsf{steer}}}

\newcommand{\union}[1]{\ensuremath{\,\langle #1 \rangle\,}}
\newcommand{\unionc}[1]{\,\langle #1 \rangle\,}

\def\sos#1#2{{\def\arraystretch{1.6}\begin{array}{c}#1\\\hline
#2\end{array}}}

\newcommand{\nar}[1]{\xrightarrow{#1}}
\newcommand{\wnar}[1]{\stackrel{#1}{\rightsquigarrow}}

\newcommand{\dnar}[1]{\stackrel{#1}{\Rightarrow}}

\newcommand{\enar}[1]{\xrightarrow{#1}_{e}}
\newcommand{\inar}[1]{\xrightarrow{#1}_{i}}

\newcommand{\pos}{{\mathbb P}}
\newcommand{\vel}{{\mathbb V}}
\newcommand{\nets}{{\mathbb N}}

\def\name#1{\mbox{\sc #1}}

\newcommand{\equaldef}{\stackrel{\mbox{\footnotesize{def}}}{=}}

\newcommand{\collision}[2]{\stackrel{{#1}}{\longleftrightarrow}_{#2}}
\newcommand{\shape}{{\sf shape}}

\newcommand{\bs}[4]{\ensuremath{\langle #1, #2, #3, \mathbf{#4}
\rangle}}
\newcommand{\boundary}[1]{\ensuremath{{\cal B}(#1)}}
\newcommand{\points}[1]{\ensuremath{{\cal P}(#1)}}
\newcommand{\velocity}[1]{{\bf v}(#1)}
\newcommand{\mass}[1]{m(#1)}
\newcommand{\referencepoint}[1]{\ensuremath{{\cal R}(#1)}}
\newcommand{\timedomain}{{\mathbb T}}

\newtheorem{definition}{Definition}
\newtheorem{theorem}{Theorem}
\newtheorem{proposition}{Proposition}

\newtheorem{lemma}{Lemma}

\begin{document}

\vspace{1cm}
\begin{center}
{\Large\bf Shape Calculus:\\ Timed Operational Semantics and Well-formedness}
\end{center}
\vspace{4mm}

\begin{center}
{\large Ezio Bartocci$^1$ \qquad
Diletta Romana Cacciagrano$^1$ \qquad
Maria \nolinebreak[4] Rita \nolinebreak[4] Di \nolinebreak[4] Berardini$^1$ \qquad
Emanuela Merelli$^1$ \qquad
Luca Tesei}\footnote{School of Science and Technology, University of Camerino\\
Via Madonna delle Carceri 9, 62032, Camerino (MC), Italy.\\ Email:
\texttt{name.surname@unicam.it}}
\end{center}

\vspace{3ex}

\date{\ }

\begin{abstract}
The Shape Calculus is a bio-inspired calculus for describing 3D shapes moving in a space. A shape forms a 3D process when combined with a behaviour. Behaviours are specified with a timed CCS-like process algebra using a notion of channel that models naturally binding sites on the surface of shapes. Processes can represent molecules or other mobile objects and can be part of networks of processes that move simultaneously and interact in a given geometrical space. The calculus embeds collision detection and response, binding of compatible 3D processes and splitting of previously established bonds. In this work the full formal timed operational semantics of the calculus is provided, together with examples that illustrate the use of the calculus in a well-known biological scenario. Moreover, a result of well-formedness about the evolution of a given network of well-formed 3D processes is proved.
\end{abstract}

\section{Introduction}

In a near future, systems biology will profoundly affect health care and medical science. One aim is to design and test ``in-silico'' drugs giving rise to individualized medicines that will take into account physiology and genetic profiles~\cite{Finkelstein2004}. The advantages of performing in-silico experiments by simulating a model, instead of arranging expensive in-vivo or in-vitro experiments, are evident. But of course the models should be as faithful as possible to the \emph{real system}.

Towards the improving of the faithfulness of languages and models proposed in literature in the field of systems biology, the Shape Calculus \cite{Bartocci2009,Bartocci2010,Bartocci2010b} was proposed as a very rich language to describe mainly, but not only, biological phenomena. The main new characteristics of this calculus are that it is spatial - with a geometric notion of a 3D space - and it is shape-based, i.e.\ entities have geometric simple or complex shapes that are related to their ``functions'', that is to say, in the context of formal calculi, the possible interactions that can occur among the entities (called 3D processes, in the same context). Note that there are some other approaches that use a geometric notion of space \cite{BarbutiMMP09,John2008,sps}, but they do not fully exploit the potentiality of geometry: they only consider positions that are centers of perception/communication spheres or discretized the space using grids. Moreover, differently to classical mathematical models for biological systems, the Shape Calculus is individual-based, that is to say it considers autonomous entities that interact with others in order to give rise to an emerging behaviour of the whole system. However, this characteristic is also present in other languages and models proposed in literature \cite{Priami2004,Regev2004,Cardelli2004,Bortolussi2007,Ciocchetta2008}. Parallel to the introduction of the Shape Calculus, a simulation environment called  {\sc BioShape} \cite{bioshapeurl,iccs2010} has been proposed. This environment embodies all the concepts and features of the Shape Calculus, plus others, more detailed, characteristics that are related to the implementation of simulations in different biological case studies \cite{iccs2010,cs2bio2010,acri2010}. The Shape Calculus can be considered the formal core of {\sc BioShape}. On the Shape Calculus, formal verification techniques can be applied, while in-silico experiments (simulations) can be performed in {\sc BioShape}. These two approaches can be used successfully in an independent way, but they can also be combined to address complex biological case studies.

In \cite{Bartocci2010}, the language of the Shape Calculus was introduced with the main aim of gently and incrementally present all its features and their relative semantics. The motivations behind the type and nature of the calculus operators were discussed and a great variety of scenarios in which the calculus may be used effectively were described. In \cite{Bartocci2010b} the timed operational semantics of the calculus was introduced and a well-formedness property of the dynamics of the calculus was proved. 

This paper is an extended version of \cite{Bartocci2010b} in which the original concepts of the calculus are fully recalled and in which further examples are provided in order to explain better the technical details of its timed operational semantics. 

We present the full timed operational semantics of the Shape Calculus and we explain with proper running examples all the technical points that need particular attention and further explanation. Moreover, we define a concept of well-formedness, starting from shapes and porting it to more complex calculus objects such as 3D processes and, ultimately, networks of 3D processes. In the calculus, only well-formed objects are considered. Well-formedness is a standard concept used to avoid strange and unwanted situations in which a term can be legally written by syntax rules, but that semantically corresponds to a contradictory situation, for instance, in our case, a composed shape whose pieces move in different directions. We prove that a given well-formed network of 3D processes always evolves into a well-formed network of 3D processes, that is to say, no temporal and spatial inconsistencies are introduced by the dynamics of the calculus.

The paper is organized as follows. Section~\ref{sec:overview} recalls the main concepts of the calculus and focuses on weak and strong split operations whose semantics are slightly more complicated to define than that of other operators. Section~\ref{sec:3ds} introduces formally 3D shapes, shape composition, movement, collision detection and collision response. Section~\ref{sec:three} defines behaviours and 3D processes giving them full semantics. Section~\ref{sec:networks} puts all the pieces together and specifies precisely networks of 3D processes and a general result of dynamic well-formedness is proven. Finally, Section~\ref{sec:conclusion} concludes with ongoing and future work directions. For the sake of readability all the proofs have been moved in Appendices~A,~B~and~C.

\section{Recall of Main Concepts}
\label{sec:overview}

\begin{figure}[t]
\begin{center}
\includegraphics[width=12cm]{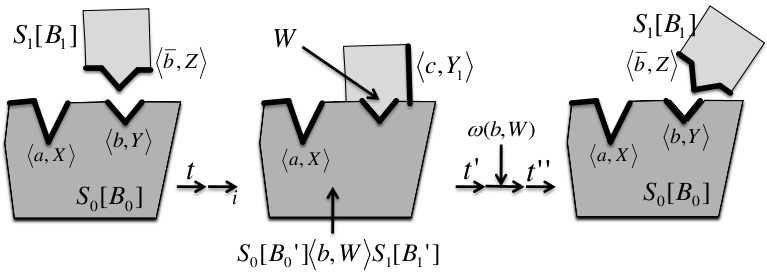}
\end{center}
\caption{An example of binding and subsequent weak-split of two 3D processes.}
\label{fig:WeakSplitting}
\end{figure}

In this section we recall the main concepts the calculus introduced in \cite{Bartocci2010} and we focus in particular to the weak and strong split operations, in order to make them clear from the beginning and, thus, then smoothly define their semantics, which requires some particular technical expedients.

The general idea of the Shape Calculus is to consider a 3D space in which several shapes reside, move and interact. While time flows, shapes move according to their velocities, that can change over time both due to a general motion law - for instance as in a gravitational or in an electromagnetic field, or as in Brownian motion - and due to collisions that can occur between two or more shapes. Collisions can result in a bounce, that is to say {\em elastic} collisions. Instead, as it often happens in biological scenarios, colliding objects can bind and become a new compound object moving in a different way and possibly having a different behaviour. In this case we speak of {\em inelastic} collisions since they are treated with the physical law for that kind of collision.

The Shape Calculus can be used to represent a lot of scenarios at different scales in different fields \cite{Bartocci2010}. However, in this work we use a well-known biological scenario in order to introduce simple running examples that illustrate the semantics of the calculus operators. We consider biochemical reactions occurring inside a cell. Every species of involved molecules has a specific shape and we know from biology that the functions of a molecule are tightly related to its shape. For instance, in enzymatic reactions the functional sites that are active in the enzyme structure, at a given time, determine which substrate (one or two metabolites) can bind the enzyme and proceed to the catalyzed reaction.

\begin{figure}[t]
\begin{center}
\includegraphics[width=12cm]{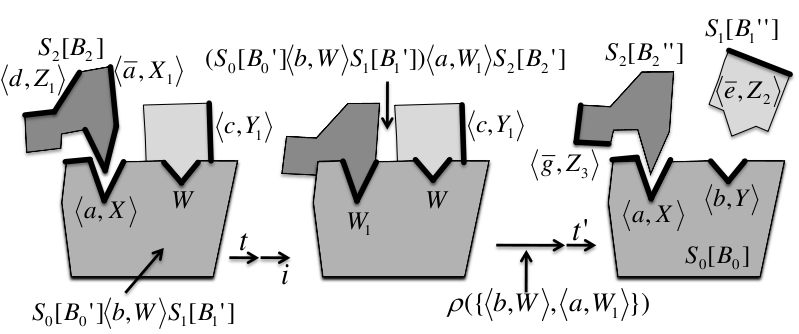}
\end{center}
\caption{An example of complex formation and subsequent strong-split.}
\label{fig:StrongSplitting}
\end{figure}

Fig.~\ref{fig:WeakSplitting} shows a (2D, for simplicity) possible representation of an enzymatic reaction. The larger object represents an enzyme with a shape $S_0$ and a behaviour $B_0$; $S_0$ and $B_0$ together constitute a 3D process $S_0[B_0]$. The 3D process $S_1[B_1]$ represents a metabolite that is close to the given enzyme. Note that portions of the shape surfaces are highlighted: they are the {\em channels} (an extension of the notion of channels of CCS \cite{Milner1989}) that the corresponding 3D processes exhibit to the environment. Each channel has a {\em name} and an {\em active surface}. For instance, $\langle a, X \rangle$ is a channel of type $a$ on the active site $X$ of $S_0[B_0]$. The enzyme in Fig.~\ref{fig:WeakSplitting} has two channels and its behaviour can be specified as: $B_0 = \langle b,Y \rangle . B'_0 + \langle a,X \rangle . B''_0$. The plus operator represents a non-deterministic choice between two possible communications on the channels. This non-determinism is resolved during the evolution of the system depending on which 3D processes will collide with the enzyme and where.

Following the evolution proposed in the figure, after some time $t$ elapsed (represented by the transition $\nar{t}$) and after a detection and resolution of an inelastic collision (transition $\inar{}$), we get one compound process represented by the term $S_0[B_0']\langle b, W \rangle S_1[B_1']$. The bond is established on the surface of contact $W$ and the name $b$ records the type of channels that bound. Note that communication, i.e.\ {\em binding}, can only happen if there is a collision between two 3D processes that expose compatible channels (name and co-name \`a la CCS) on their surface of contact. If the channels were not compatible, the collision would have been treated as elastic and the two 3D processes would have simply bounced. By letting $B'_1 = \langle c, Y_1 \rangle. B''_1$ we allow the component of shape $S_1$ to open a new channel and to bind with other colling processes.

The third stage of Fig.~\ref{fig:WeakSplitting} represents a possible evolution of an enzyme binding with a substrate; it can happen that, for some reason, the bond is loose and the two molecules are free again. To model this kind of behaviour we introduce a special kind of {\em not urgent} split operation, called {\em weak-split operation}, that can be delayed of an unspecified time. This is another source of non-determinism in the calculus. Fig.~\ref{fig:StrongSplitting} shows what happens when another substrate -- that we represent as the 3D process $S_2[B_2]$ with a channel $\langle \overline{a}, X_1 \rangle$  -- binds with the compound process $S_0[B_0']\langle b, W \rangle S_1[B_1']$ on the common surface $W_1$. In the terminology of biochemical reactions, a final complex has been formed. As a consequence, the reaction must proceed and the products must be released. In our calculus, reactions are represented by  {\em strong-split operations}. Differently from weak-split operations, this kind of split operations cannot be delayed and {\em must} occur as soon as they are enabled, i.e.\ when all the involved components can release all the bonds. In this example the involved components are $S_0[B'_0]$, $S_1[B'_1]$ and $S_2[B'_2]$; the set of bonds to be split is $L= \{\langle b, W \rangle, \langle a, W_1 \rangle \}$.

\section{3D Shapes}
\label{sec:3ds}

We start by introducing three dimensional shapes as terms of a suitable language, allowing simpler shapes to bind and form more complex shapes. From now on we consider assigned a {\em global} coordinate system in a three dimensional space represented by $\real^3$. Let $\pos, \vel = \real^3$ be the sets of positions and velocities, respectively, in this coordinate system. Throughout the paper, we assume relative coordinate systems that will always be w.r.t.\ a certain shape $S$, i.e.\ the origin of the relative system is the reference point ${\bf p}$ of $S$. We refer to this relative system as the {\em local} coordinate system of shape $S$. If $\mathbf{p} \in \pos$ is a position expressed in global coordinates, and $V \subseteq \pos$ is a set of points expressed w.r.t.\ a local coordinate system whose origin is $\mathbf{p}$, we define $\reference(V, {\bf p}) = V + {\bf p} = \{ {\bf x} + {\bf p} \,|\, {\bf x} \in V\}$ to be the set $V$ w.r.t.\ the global coordinates. Using local coordinate systems we can express parts of a given shape, such as faces and vertexes, independently from its actual global position.

\begin{definition}[Basic Shapes]
\label{def:bshapes}
A basic shape $\sigma$ is a tuple $\bs{V}{m}{\mathbf{p}}{v}$ where $V \subseteq \pos$ is a {\em convex polyhedron} (e.g.\ a \emph{sphere}, a \emph{cone}, a \emph{cylinder}, etc.)\footnote{From a syntactical representation point of view, we assume that $V$ is finitely represented by a suitable data structure, such as a formula or a set of vertices.} that represents the set of shape points; $m \in \real^+$, ${\bf p} \in \pos$ and ${\bf v} \in \vel$ are, resp., the mass, the centre of mass\footnote{We actually need only a reference point. Thus, any other point in $V$ can be chosen.} and the velocity of $\sigma$. All possible basic shapes are ranged over by $\sigma, \sigma', \cdots$. We also define the {\em boundary} $\boundary{\sigma}$ of $\sigma$ to be the subset of points of $V$ that are on the surface of $\sigma$\footnote{Note that we consider only closed shapes, i.e.\ they contain their boundary.}
\end{definition}

Note that basic shapes are very simple convex shapes. They can be represented by means of suitable and efficient data structures and are handled by the most popular algorithms for motion simulation, collision detection and response~\cite{Ericson2005}. Three dimensional shapes of any form can be approximated - with arbitrary precision - by composing basic shapes: the composition of two shapes corresponds to the construction of a third shape by ``gluing the two components on a common surface. This concept is generalized by the following definition.

\begin{definition}[3D shapes]
\label{def:shapes}
The set $\shapes$ of the  {\em 3D shapes}, ranged over by $S, S', \cdots$, is generated by the grammar:
$$S ::= \sigma \gor  S \union{X} S $$
where $\sigma$ is a basic shape and $X \subseteq \pos$. If $S = \sigma = \bs{V}{m}{\mathbf{p}}{v}$, we define $\points{S} = V$, $\mass{S} = m$, $\referencepoint{S} = {\bf p}$,  $\velocity{S} = \{ {\bf v}\}$ to be, resp., the {\em set of points}, the {\em mass}, the {\em reference point} and the {\em velocity} of $S$. If $S = S_1 \union{X} S_2$  is a compound shape, then: $\points{S} = \points{S_1} \cup \points{S_2}$, $\mass{S} = \mass{S_1} + \mass{S_2}$, $\referencepoint{S} = \big(\mass{S_1}\cdot \referencepoint{S_1} + \mass{S_2} \cdot \referencepoint{S_2}\big) / \big( \mass{S_1} + \mass{S_2} \big)$\footnote{Again for simplicity, we use the centre of mass as the reference point. Any other point can also be chosen.} and $\velocity{S} = \velocity{S_1} \cup \velocity{S_2}$. Finally, the boundary of $S$ is defined to be the set $\boundary{S} = \big( \boundary{S_1} \cup \boundary{S_2}\big) \setminus \{ \mathbf{x} \in \pos \mid \mathbf{x} \mbox{ is interior of } \points{S_1 \union{X} S_2} \}$, where a point $\mathbf{x} \in V \subseteq \pos$ is said to be \emph{interior} if there exists an open ball with centre $\mathbf{x}$ which is completely contained in $V$.
\end{definition}

In this paper we only consider shapes that are \emph{well-formed} according to the following
definition.

\begin{definition}[Well-formed shapes]
\label{def:shapeswf}
Each basic shape $\sigma$ is well-formed. A compound shape $S_1 \union{X} S_2$ is well-formed
iff:
\begin{enumerate}
\item both $S_1$ and $S_2$ are well-formed;
\item \label{item:boundaries} the set $X = \points{S_1} \cap \points{S_2}$ is {\em non-empty} and
equal to $\boundary{S_1} \cap \boundary{S_2}$;
\item \label{item:velocity} $\velocity{S_1 \union{X} S_2}$ is a singleton $\{{\bf v}\}$%
\footnote{With abuse of notation, throughout the paper, we write $\velocity{S}$ also to refer to the element $\mathbf{v}$ of the singleton  $\{ \mathbf{v} \}$ as this is not ambiguous when $S$ is
well-formed.} where ${\bf v} = \velocity{S_1} = \velocity{S_2}$.
\end{enumerate}

\noindent Below we say that two shapes $S_1$ and $S_2$ {\em interpenetrate} each other if there exists a point $\mathbf{x}$ that is interior of both $\points{S_1}$ and $\points{S_2}$. In other terms, they interpenetrate iff the set $X= \points{S_1} \cap \points{S_2} \neq \emptyset$ is not a subset of $ \boundary{S_1} \cap \boundary{S_2}$. If $S = S_1 \union{X} S_2$ is well-formed, $X$ is said to be the {\em surface of contact} between $S_1$ and $S_2$; moreover, each  $\mathbf{x} \in X$ is a {\em point of contact}.
\end{definition}

Condition~(\ref{item:boundaries}) guarantees that well-formed compound shapes touch but do not
interpenetrate; the surface of contact $X$ is always on the boundary of both $S_1$ and $S_2$. It
can be a single point, a segment or a surface, depending on where the two shapes are touching
without interpenetrating. Most of the time $X$ is a (subset of) a \emph{feature} of the
basic shapes composing the 3D shape, i.e., a face, an edge or a vertex.
Condition~(\ref{item:velocity}) imposes that all the shapes forming a compound shape have the same
velocity; thus, the compound shape moves as a unique body.

A compound 3D shape $S$ can be represented in a number of different ways by rearranging its basic
shapes and surfaces of contact. All these possible representation are `equivalent' w.r.t.\  the
structural congruence defined below.

\begin{definition}[Structural Congruence of 3D Shapes]
\label{def:shapecongruence}
The structural \\ congruence relation over $\shapes$, denoted by $\equiv_S$, is the smallest relation that satisfies the following rules:
\begin{enumerate}
\item $S_1 \union{X} S_2 \equiv_S S_2 \union{X} S_1$;
\item $(S_1 \union{X} S_2) \union{Y} S_3 \equiv_S S_1 \union{X} (S_2 \union{Y} S_3)$ provided that
the surface of contact $Y \subseteq \boundary{S_2} \cap \boundary{S_3}$.
\end{enumerate}
\end{definition}

The next example shows how the particular features of Shape Calculus can be used to construct a model in a well-known biological scenario.

\begin{example}[A Biological Example]
\label{sec:example}
The \emph{glycolysis} pathway is part of the process by which individual cells produce and consume
nutrient molecules. It consists of ten sequential reactions, all catalyzed by a specific enzyme. We focus, in this example, on the first reaction that can be described as

\emph{glucose, ATP  \,\, \ce {<=>} \,\, glucose-6-phosphate, ADP, H$^{+}$}

\noindent where an ATP is consumed and a molecule of glucose (GLC) is phosphorylated to glucose 6-phosphate (G6P), releasing an ADP molecule and a positive hydrogen ion (Hydron).

\begin{figure}[tbh]
\begin{center}
\includegraphics[width=12cm]{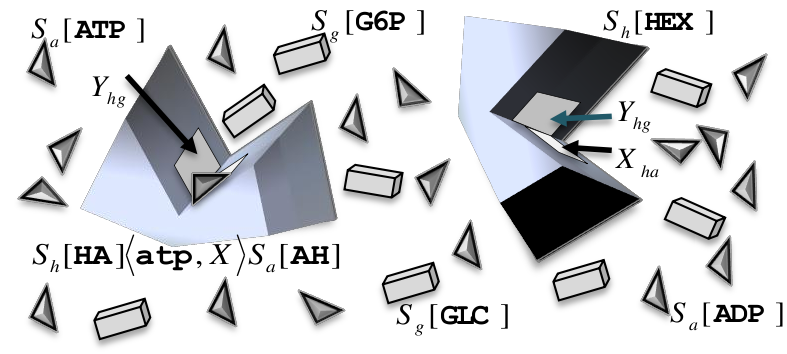}
\end{center}
\caption{A network of 3D processes for describing the first reaction of the glycolysis pathway.}
\label{fig:glyco}
\end{figure}

The enzyme catalyzing this first reaction is \emph{Hexokinase} (HEX). GLC, G6P, ATP, ADP and H${^+}$ are metabolites. Both enzymes and metabolites are autonomous cellular entities that continuously move within the cytoplasm. The transformation of a metabolite into the one that follows in the ``pipeline'' of the pathway (in this case, GLC into G6P) depends on the meeting (collision in binding sites) of the right enzyme (in this example HEX) with the right metabolites, in this example GLC and ATP. The order of these bindings does not matter. After this binding the reaction takes place and the products\footnote{In this example we neglect the hydron.} are released in the cytoplasm. A special case occurs when the enzyme has bound one metabolite and an environmental event makes it release the metabolite and not proceed to the completion of the reaction.

We model the shape of HEX, which we denote with $S_h$, by a polyhedron approximating its real shape and mass, available at public databases (e.g.\ \cite{PDB}). Fig.~\ref{fig:glyco} shows a network of 3D processes in which there are two hexokinases and some GLC, G6P, ATP and ADP 3D processes. Note that we use a unique kind of shape for GLC ad G6P, denoted by $S_g$, and another unique kind of shape for ATP and ADP, denoted $S_a$. They will be distinguished by their behaviours.

Note that modelling of biochemical reactions with our calculus is very different from the usual ODE-based approach. This is because the Shape Calculus embeds concepts and features of a particle-based, individual-based and spatial geometric approach and we wanted to show, in a relatively simple and well-known scenario, how all these concepts and features can be used.
\end{example}

\subsection{Trajectories of Shapes}
\label{sec:trajectories}
The general idea of the Shape Calculus is to consider a three-dimensional space in which several shapes reside, move and interact. One of the choices to be made is how the velocity of each shape changes over time. We believe that a continuous updating of the velocity - that would be a candidate for an `as precise as possible approach of modelling - is not a convenient choice. The main reason is the well-known compromise between the benefits of approximation and the complexity of precision. Our choice, quite common also in computer graphics~\cite{Ericson2005}, is to approximate a continuous trajectory of a shape with a polygonal chain, i.e.\ a piecewise linear curve in which each segment is the result of a movement with a constant velocity. The vertices of the chain corresponds to velocity updates.

To this purpose we define a global parameter $\mts \in \real^+$, called \emph{movement time step},
that represents the maximum period of time after which the velocities of all shapes are updated.
The quantification of $\mts$ depends on the desired degree of approximation and also on other
parameters connected to collision detection (see Section~\ref{sec:collisiondetection}). In some situations, the time of updating can be shorter than $\mts$ because, before that time, collisions between moving shapes can occur. These collisions must be resolved and the whole system must re-adapt itself to the new situation. The time domain $\timedomain = \real^+_0$ is then divided into an infinite sequence of time steps $t_i$ such that $t_0=0$ and $t_i \leq t_{i-1} + \mts$ for all $i > 0$. An example in~\cite{Bartocci2010} (cf. Section 2) shows in more details how the timeline can be broken up into such time instants.

In the calculus, the updating of velocities is performed by exploiting a function $ \move \colon \timedomain \rightarrow (\shapes \hookrightarrow \vel)$ that describes how the velocity of all existing shapes (i.e.\ all shapes that are currently moving in the space) at each time $t$ is changed. We assume that, at any given time instant $t \in \timedomain$, $\move \, t \,S $ is undefined iff shape $S$ does not exist and, hence, its velocity has not to be changed.

This approach provides us with the maximal flexibility for defining motion. Static shapes are those shapes whose velocity is always zero\footnote{To represent walls, we also need to assign an infinite value to the mass of these objects, otherwise they can be moved anyway due to collisions.}. A gravity field can be simulated by updating the velocities according to the gravity acceleration. A Brownian motion can be simulated by choosing a random 3D direction and then defining the length of the vector w.r.t.\ the mass and/or the volume of the shape. In this paper, we do not consider movements due to rotations. However, this kind of movements can be easily added to our shapes by assigning an angular velocity and a moment of inertia to a shape and then by performing a compound motion of translation and rotation along the movement time step.

Let us define now some useful notation and properties.

\begin{definition}[Evolution of shapes over time]
\label{def:shapemoving}
Let $S \in \shapes$ and $t \in \timedomain$; $S + t$, i.e.\ the shape $S$ after $t$ time
units, is defined by induction on $S$:
\begin{tabbing}
Basic: \quad \= $\bs{V}{m}{{\bf p}}{v} + t = \bs{V + (t \cdot {\bf v})}{m}{\mathbf{p} + (t \cdot
{\bf v})}{v}$\\
Comp: \> $(S_1 \union{X} S_2) + t = (S_1 + t) \union{X + t\cdot \velocity{S}}(S_2 + t)$
\end{tabbing}
\end{definition}

\begin{definition}[Updating shape velocity]
\label{def:shape-velocity}
Let $S\in \shapes$ and $\vect{w}\in\vel$. We define the shape $S\lsbrace\vect{w}\rsbrace$, i.e.\
$S$ whose velocity is updated with $\vect{w}$, as follows:

\begin{tabbing}
Basic: \quad \= $\bs{V}{m}{{\bf p}}{v} \lsbrace\vect{w}\rsbrace = \bs{V}{m}{{\bf p}}{w}$\\
Comp: \> $(S_1 \union{X} S_2) \lsbrace\vect{w}\rsbrace = (S_1 \lsbrace\vect{w}\rsbrace)
\union{X}(S_2\lsbrace\vect{w}\rsbrace)$
\end{tabbing}
\end{definition}

The following result comes directly from Def.~\ref{def:shapeswf}.

\begin{proposition}
\label{prop:shapes-wellformedness}
Let $S \in \shapes$, $t\in \timedomain$ and $\vect{w}\in\vel$. If $S$ is
well-formed then $S+t$ and $S\lsbrace\vect{w}\rsbrace$ are well-formed.
\end{proposition}

Our intent is to represent a lot of shapes moving simultaneously in space as described above.
Inevitably, this scenario produces collisions between shapes when their trajectories encounter.

\subsection{Collision Detection}
\label{sec:collisiondetection}
There is a rich literature on collision detection systems, as this problem is fundamental in popular applications like computer games. Good introductions to existing methods for efficient collision detection are available and we refer to Ericson \cite{Ericson2005} and references therein for a detailed treatment.

For our purposes, it is sufficient to define an interface between our calculus and a typical collision detection system. We can then choose one of them according to their different characteristics, e.g.\ their applicability in large-scale environments or their robustness. It must be said, however, that the choice of the collision detection system may influence the kind of basic (or compound) shapes we can use, as, for instance, some systems may require the use of only convex shapes to be more efficient\footnote{The basic shapes that we consider in Definition~\ref{def:bshapes} are typically accepted by most of the collision detection systems.}.

\begin{figure}
\begin{center}
\includegraphics[width=12cm]{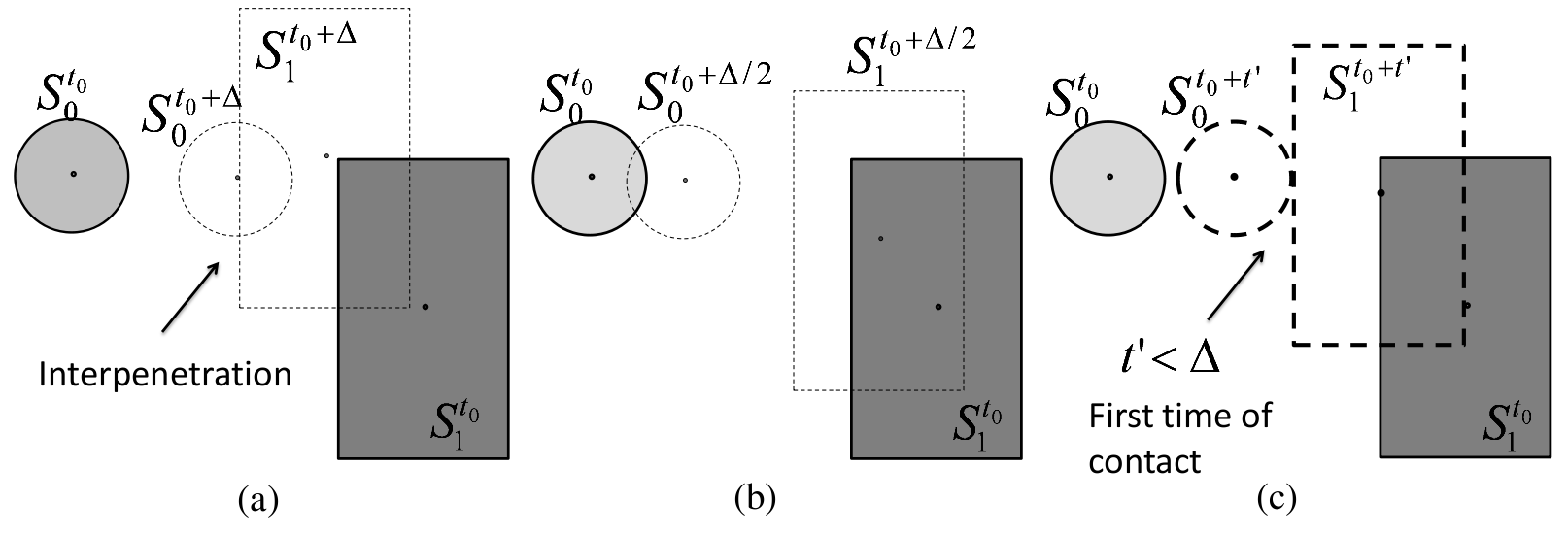}
\end{center}
\caption{Some steps to determine the first time of contact between two shapes.}
\label{fig:ftc}
\end{figure}

Typically, a collision detection algorithm assumes to start in a situation in which shapes do not interpenetrate. Then it tries to move all the shapes of a little\footnote{The time step must be chosen little enough to avoid that any two shapes, at their possible maximum velocity, start in a non-interpenetration state, engage in an interpenetration and exit from the interpenetration state in a single time step duration.} time step - that we have already introduced as the movement time step $\mts$ - and check if interpenetrations occurred\footnote{Typically, the major efforts of optimization are concentrated in this step since the number of checks is, in the worst case, $O(N^2)$ - where $N$ is the number of shapes in the space - but the shapes that are likely to collide are only those that are very close to each other.}. If so, it tries to consider only half of the original time step and repeat the interpentration check, i.e.\ it performs a binary search of the \emph{first time of contact} $t$ between two or more shapes, with some degree of approximation. Fig.~\ref{fig:ftc} shows these steps. In case $(a)$ the passage of the whole $\mts$ results in an interpenetration. Then, in $(b)$ the passage of $\mts/2$ is tried resulting into a non-contact. After some iterations the situation in $(c)$ is reached.

In addition to the first time of contact, a collision detection algorithm usually outputs the shapes that are colliding, i.e.\ are touching without interpenetrating after $t$, and some information about the surfaces or points of contact. We now define precisely what we expect to obtain from a collision detection system.

\begin{definition}[First time of contact]
\label{def:ftoc}
Let $I$ be a non-empty finite set of indexes and let $\{S_i\}_{i \in I}$ be a set of well-formed
shapes such that for all $i,j \in I$, $S_i$ and $S_j$ do not interpenetrate (see Def.~\ref{def:shapeswf}).
The \emph{first time of contact} of the shapes $S_i$, denoted $\ftoc(\{S_i\}_{i \in I})$, is a
number $t\in \timedomain$ such that:
\begin{enumerate}
\item for all $t' \in \timedomain$, $0 \leq t' \leq t$ and for all $i,j \in I$, $S_i+t'$ and
$S_j+t'$ do not interpenetrate;

\item there exist $i,j \in I$, with $i\neq j$, such that $\boundary{S_i + t} \cap \boundary{S_j + t} \neq \emptyset$, i.e., some shapes are touching at $t$;

\item for all $\epsilon > 0$ there exists $\delta$, $0 < \delta < \epsilon$, and $i,j \in I$, with $i \neq j$, such that $S_i+ (t+\delta)$ and $S_j+(t+\delta)$ interpenetrate, i.e., in $t$ some shapes are touching and any further movement makes them to interpenetrate.
\end{enumerate}
\end{definition}

Note that such a definition allows shapes that are touching without interpenetrating, and with
velocities that do not make them to interpenetrate (e.g., the same velocity), to move without
triggering a first time of contact. This  will be useful in Section~\ref{sec:networks} when we split previously compound shapes. Indeed, after the split these shapes have the same velocity and, hence, do not affect the next first time of contact.

\begin{definition}[Collision information]\label{def:colls-information}
Let $\{S_i\}_{i \in I}$ be a set of well-formed shapes and let $t = \ftoc(\{S_i\}_{i \in I})$ be
their first time of contact. The {\em set of colliding shapes} after time $t$  is denoted by
$\colliding(\{S_i\}_{i \in I}) \subseteq \shapes \times \shapes \times \wp(\pos)$. A tuple
$\langle S_i,S_j, X \rangle \in \colliding(\{S_i\}_{i \in I})$ iff:

\begin{enumerate}

\item $\points{S_i+t} \cap \points{S_j+t}$ is \emph{non-empty} and is equal to $\boundary{S_i + t}
\cap \boundary{S_j + t}$;

\item for all $\epsilon > 0$ there exists $\delta$, $0 < \delta < \epsilon$, such that $S_i+
(t+\delta)$ and $S_j+(t+\delta)$ interpenetrate.
\end{enumerate}
\end{definition}

\subsection{Collision Response} \label{sec:collisionresponse}

In this section, we briefly discuss the problem of {\em collisions response}~\cite{Hecker1997},
i.e.\ how collisions, once detected, can be resolved. We distinguish between {\em elastic} collisions (those in which there is no loss in kinetic energy) and {\em perfectly inelastic} ones (in which kinetic energy is fully dissipated)\footnote{Other different kinds of collisions can be easily added to the calculus provided that the corresponding collision response laws are given.}. After an elastic collision, the two shapes will proceed independently to each other but their velocities will be changed according to the laws for conservation of linear momentum and kinetic energy - Equations (1)-(2) in Def.~\ref{def:collisions}. On the contrary, two shapes that collide inelastically will bind together and will move as a {\em unique} body whose velocity is given by the law for conservation of linear momentum only - Equation~(3), in Def.~\ref{def:collisions}.

\begin{definition}[Collisions]
\label{def:collisions}
Let $S_1, S_2\in \shapes$ and let $X \subseteq \pos$ be a surface of contact between them. If $X$ is neither an edge nor a vertex of $S_1$, the velocities $\vect{w}_1$ and $\vect{w}_2$ of these shapes after an elastic collision in $X$ are given by:
$$
\begin{array}{c c c c }
(1) & \vect{w}_1 = \velocity{S_1}  - \dfrac{\lambda}{\mass{S_1}} \cdot \vect{n}
\quad  &
(2) & \vect{w}_2 = \velocity{S_2}  +  \dfrac{\lambda}{\mass{S_2}} \cdot \vect{n} \\
\end{array}
$$
\noindent where $\vect{n}$ is the normal of contact away from $X \subseteq \boundary{S_1}$, i.e.\
the unit vector perpendicular to the face of $S_1$ that contains $X$, and%
$$\lambda = 2 \, \dfrac{\mass{S_1} \, \mass{S_2}}{\mass{S_1} + \mass{S_2}}
\, \dfrac{\velocity{S_1} \cdot \vect{n} - \velocity{S_2} \cdot \vect{n}} {\vect{n} \cdot
\vect{n}}$$
If $X$ is either an edge or a vertex of $S_1$, $\vect{n}$ is the normal of contact away from the
shape $S_2$ and velocities $\vect{w}_1$ and $\vect{w}_2$ are obtained by means of
symmetric equations. In both cases, we write $S_1 \collision{X}{\mathit{e}} S_2$ to denote the pair
of velocities $(\vect{w}_1, \vect{w}_2)$. If $S_1$ and $S_2$ collide inelastically in the surface of
contact $X$, they will bind together as a unique body whose velocity (denoted with
$S_1 \collision{X}{\mathit{i}} S_2$) is given by:
$$\begin{array}{c c}
(3) & \vect{v} = \dfrac{\mass{S_1}}{\mass{S_1} + \mass{S_2}} \cdot
\velocity{S_1} +
\dfrac{\mass{S_2}}{\mass{S_1} + \mass{S_2}} \cdot \velocity{S_2}\\
\end{array}$$
\end{definition}

\section{3D Processes}
\label{sec:three}
In this section we introduce the timed process algebra whose terms describe the {\em internal
behaviour} of our 3D shapes. This is a variation of Timed CCS~\cite{Yi1990}, where basic actions
provide information about binding capability and splits of shape bonds. We use the following
notation. Let $\Lambda =\{a, b, \cdots\}$ be a countably infinite set of {\em channels names}
(names, for short) and $\overline{\Lambda} = \{\overline{a} \,|\, a \in \Lambda\}$ its
complementation; by convention we assume $\overline{\overline{a}} = a$ for each name $a$. Elements
in $\cname = \Lambda \cup \overline{\Lambda}$ are ranged over by $\alpha,\beta, \cdots$.

Binding capabilities are represented by {\em channels}, i.e.\ pairs $\langle \alpha, X\rangle$ where $\alpha \in \cname$ is a name and $X$ is a {\em surface of contact}. Intuitively, a surface
of contact is a portion of space (usually, a subset of the boundary of a given 3D shape) where the
channel itself is active and where binding with other processes are possible. Names introduce a
notion of compatibility between channels useful to distinguish between elastic and inelastic
collisions. If $\beta = \overline{\alpha}$ and $X \cap Y\neq \emptyset$ we say that the channels
$\langle \alpha, X \rangle$ and $\langle \beta, Y\rangle$ are {\em compatible}, written  $\langle
\alpha, X \rangle \comp \langle \beta, Y \rangle$. Otherwise, $\langle \alpha, X \rangle$ and
$\langle \beta, Y \rangle$ are {\em incompatible} which we write as $\langle \alpha, X \rangle \not \comp \langle \beta, Y \rangle$.

We also introduce two  different kinds of actions that represent splits of shape bonds. In particular, we distinguish between {\em weak-split actions} of the form $\omega(\alpha, X)$ and {\em strong-split actions} of the form $\rho(\alpha,X)$. With an abuse of notation, two weak-split actions $\omega(\alpha, X)$ and $\omega(\beta,Y)$ (as similarly for the strong-split actions $\rho(\alpha, X)$ and $\rho(\beta,Y)$) are compatible if so are the channels $\langle \alpha, X \rangle$ and $\langle \beta, Y \rangle$. We will see that a synchronization between a pair of compatible weak-split actions results in a weak-split operation, while synchronizations between multiple pairs of compatible strong-split actions correspond to a strong-split operation. These operations behave differently w.r.t.\ to {\em time passing} since the latter operation cannot time pass further, while the former one can be arbitrarily delayed.

Let $\channels$ be the set of all channels, $\omega(\channels) = \{\omega(\alpha, X) \,|\, \langle \alpha, X \rangle \in \channels\}$ and $\rho(\channels) = \{\rho(\alpha, X) \,|\, \langle \alpha, X \rangle \in \channels\}$ be the sets of {\em weak-split actions} and {\em strong-split actions}, resp. Our processes perform basic and {\em atomic} actions that belong to the set $\Act = \channels \cup \omega(\channels) \cup \rho(\channels)$ whose elements are ranged over by $\mu, \mu',\cdots$. We finally assume a countably infinite collection $\pconst$ of {\em process name} or {\em process constants}.

\begin{definition}[Shape behaviours]
\ \\ The set of {\em shape behaviours}, denoted by  $\bpa$, is generated by the following grammar:
$$B ::=\nil \gor \langle \alpha, X \rangle.B \gor \omega(\alpha,X).B \gor  \rho(L). B \gor
\epsilon(t).B \gor B + B \gor K $$
where $\langle \alpha, X\rangle \in \channels$, $L \subseteq \channels$ (non-empty) whose elements
are pairwise incompatible (i.e.\ for each pair $\langle \alpha, X \rangle, \langle \beta, Y \rangle
\in L$ it is $\langle \alpha, X \rangle \not\comp \langle \beta, Y \rangle$), $t\in \timedomain$ and
$K$ is a process name in $\pconst$.
\end{definition}

A brief description of our operators now follows. $\nil$ is the Nil-behaviour, it can not perform any action but can let time pass without limits. A trailing $\nil$ will often be omitted, so e.g.\ we write $\langle a, X \rangle. \omega(a, X)$ to abbreviate $\langle a, X \rangle. \omega(a, X).\nil$. $\langle \alpha, X \rangle.B$ and $\omega(\alpha, X).B$ are (action-)prefixing known from CCS; they evolve in $B$ by performing the actions $\langle \alpha, X \rangle$ and $\omega(\alpha, X)$, resp. $\langle \alpha, X \rangle. B$ exhibits a binding capability along the channel $\langle \alpha, X \rangle$, while $\omega(\alpha, X).B$ models the behaviour of a shape that, before evolving in $B$, wants to split a {\em single} bond established via the channel $\langle \alpha, X \rangle$.
$\rho(L).B$ is the {\em strong-split operator}; it can evolve in $B$ only after the execution of
{\em all} strong-split actions $\rho(\alpha, X)$ with $\langle \alpha, X \rangle \in L$.
The delay-prefixing operator $\epsilon(t).B$ (see~\cite{Yi1990}) introduces time delays in 3D processes; $t \in \timedomain$ is the amount of time that has to elapse before the idling time is over (see rules \name{Delay$_t$} and \name{Delay$_a$} in Tables~\ref{table:beh-temporal} and~\ref{table:beh-functional}). Finally, $B_1 + B_2$ models a non-deterministic choice between $B_1$ and $B_2$ and $K$ is a process definition.

In the remainder of this paper, we use processes in $\bpa$ to define the internal behaviour of
3D shapes. For this reason, we assume that sites in binding capabilities, as well as in weak- and
strong-split actions, are expressed w.r.t.\ the {\em local} coordinate system whose origin is
the reference point of the shape where they are embedded in.

\begin{definition}[Operational semantics of shape behaviours]
\label{def:behaviours-semantics}
The\\
\noindent SOS-rules that define the {\em temporal transition} relations $\wnar{t} \, \subseteq \,
(\bpa \times \bpa)$ for $t \in \timedomain$, that describe how shape behaviours evolve by letting
time $t$ pass, are provided in Table~\ref{table:beh-temporal}. We write $B \wnar{t} B'$ if
$(B,B')\in \wnar{t}$ and $B \wnar{t}$ if there is $B'\in \bpa$ such that $(B,B')\in \wnar{t}$.
Similar conventions will apply later on.
Rules in Table~\ref{table:beh-functional} define the {\em action transition} relations $\nar{\mu}
\subseteq (\bpa \times \bpa)$ for $\mu \in \Act$. These transitions describe which basic actions a
shape behaviour can perform.
\end{definition}

{\small
\begin{table}
\[
\begin{array}{|c|}
\hline
\name{Nil$_t$}\sos{}{\nil \wnar{t} \nil}
\quad
\name{Pref$_t$} \sos{ \mu \in \channels \cup \omega(\channels)}
{\mu.B \wnar{t} \mu.B}
\quad
\name{Str$_t$} \sos{} {\rho(L).B \wnar{t} \rho(L).B} \\
\name{Sum$_t$} \sos{B_1 \wnar{t} B'_1 \quad B_2 \wnar{t} B'_2}
{B_1 + B_2 \nar{t} B'_1 + B'_2}
\quad
\name{Del$_t$} \sos{t' \geq t} {\epsilon(t').B \wnar{t} \epsilon(t' - t).B}  \\
\name{Def$_t$} \sos{B \wnar{t} B'}
{K \wnar{t} B'} \quad \mbox{ if } K \equaldef B\\
\hline
\end{array}\]
\caption{Temporal behaviour of $\bpa$'s terms}
\label{table:beh-temporal}
\end{table}
}

Most of the rules in Table~\ref{table:beh-temporal} are those provided in~\cite{Yi1990}.
Rules~\name{Pref$_t$} and~\name{Str$_t$} state that processes like $\langle \alpha, X \rangle.B$,
$\omega(\alpha,X).B$ and $\rho(L).B$ can be arbitrarily delayed. The only rules in
Table~\ref{table:beh-functional} worth noting are those defining the functional behaviour of the
strong-split operator.
By Rules~\name{Str$_{1}$} and~\name{Str$_{2}$}, if  $\langle \alpha,X \rangle \in L$ then
$\rho(L).B$ can do a $\rho(\alpha,X)$-action and evolve either
in $B$ (if $L=\{\langle \alpha,X \rangle\}$) or in $\rho(L\backslash \{\langle \alpha,X
\rangle\}).B$ (otherwise). Rule~\name{Str$_{3}$} is needed to handle arbitrarily nested terms,
e.g.\ $\rho(\{\langle a, X \}).\rho(\{\langle \overline{b}, Y \}).B$. Other rules are as expected.

{\small
\begin{table}[t]
\[
\begin{array}{|c|}
\hline
\name{Pref$_{a}$} \sos{\mu \in \channels \cup \omega(\channels)}{%
\mu.B \nar{\mu} B
} \quad
\name{Del$_a$} \sos{ B \nar{\mu} B'}{%
\epsilon(0).B \nar{\mu} B'
} \quad
\name{Sum$_a$} \sos{B_1 \nar{\mu} B'}{%
B_1 + B_2 \nar{\mu} B'}  \\
\name{Def$_a$} \sos{B \nar{\mu} B'} {%
K \nar{\mu} B'} \quad  \mbox{ if } K \equaldef B\qquad
\name{Str$_{1}$} \sos{L=\{\langle \alpha, X \rangle\}}{%
\rho(L).B \nar{\rho(\alpha, X)} B
}
\\
\name{Str$_{2}$} \sos{ L = \{\langle \alpha, X \rangle\} \cup L' \;\; L' \neq\emptyset}{%
\rho(L).B \nar{\rho(\alpha, X)} \rho(L').B
} \quad
\name{Str$_{3}$} \sos{B \nar{\rho(\alpha, X)} B'}{%
\rho(L).B \nar{\rho(\alpha, X)} \rho(L).B'
}
  \\
\hline
\end{array}\]
\caption{Functional behaviour of $\bpa$-terms}
\label{table:beh-functional}
\end{table}
}

Now we are ready to define 3D processes, i.e.\ simple or compound shapes whose behaviour is expressed by a process in $\bpa$.

\begin{definition}[3D processes]
\label{def:processes}
The set $\proc$ of {\em 3D processes} is generated by the following grammar:
$$P:: = S[B] \gor P \unionc{a,X} P$$
where $S\in \shapes$, $B \in \bpa$, $a \in \Lambda$ and $X$ is a non-empty subset of $\pos$. The
shape of each $P\in \proc$ is defined by induction on $P$ as follows:

\begin{tabbing}
Basic: \quad \= $\shape(S[B]) = S$\\
Comp: \> $\shape(P \unionc{a,X} Q) = \shape(P) \union{X} \shape(Q)$
\end{tabbing}
We also define $\velocity{P} = \velocity{\shape(P)}$ and $\boundary{P} =\boundary{\shape(P)}$. Below
we often write $P \collision{X}{i}Q$  and $P \collision{X}{e} Q$ as  shorthand for $\shape(P)
\collision{X}{i} \shape(Q)$ and $\shape(P) \collision{X}{e} \shape(Q)$, resp. Finally, $P \lsbrace
\vect{v} \rsbrace$ is the 3D process we obtain by updating $P$'s velocity as follows:
\begin{tabbing}
Basic: \quad \=  $(S[B])\lsbrace \vect{v} \rsbrace = (S\lsbrace \vect{v} \rsbrace)[B]$\\
Comp: \> $(P \unionc{a,X} Q)\lsbrace \vect{v} \rsbrace = (P\lsbrace \vect{v} \rsbrace)
\unionc{a,X} (Q\lsbrace \vect{v} \rsbrace)$
\end{tabbing}
We finally write $\move\ t \, P$ to denote $P \lsbrace \move\ t \, \shape(P) \rsbrace$. We say that
a basic process $S[B]$ is {\em well-formed} iff the shape $S$ is well-formed and, for each $X
\subseteq \pos$ that occurs in $B$, \ $\reference(X, \referencepoint{S} ) \subseteq \boundary{S}$. A
compound process $P \unionc{a,X} Q$ is well-formed iff $P$ and $Q$ are well-formed, $\velocity{P} =
\velocity{Q}$ and the site $X$ expressed w.r.t a global coordinate system is a non-empty subset of
$\boundary{P} \cap \boundary{Q}$. Note that this also means that the set $\points{P} \cap
\points{Q}$ is non-empty and equal to $\boundary{P} \cap \boundary{Q}$. Later on in this paper we
only consider well-formed processes.
\end{definition}

We can state the following proposition as an easy consequence of shapes and 3D processes
well-formedness.

\begin{proposition}
\label{prop:well-formed}
For each $P \in \proc$ well-formed, $\shape(P)$ is well-formed.
\end{proposition}

Let us model the molecules involved in the reaction of~Example~\ref{sec:example} as 3D
processes.

\begin{example}[3D Processes for $HEX$, $GLC$ and $ATP$]\label{ex:running}
An Hexokinase molecule is modeled as $HEX \eqdef S_h[\mathsf{HEX}]$ where:

\smallskip

\noindent $\mathsf{HEX} = \langle \mathsf{atp}, X_{ha} \rangle.\mathsf{HA} + \langle
\mathsf{glc}, X_{hg} \rangle.\mathsf{HG}$,

\noindent $\mathsf{HA} =  \omega({\sf atp}, X_{ha}) . {\sf HEX} + \epsilon(t_h). \langle {\sf glc},
X_{hg} \rangle. \rho(\{ \langle {\sf atp}, X_{ha} \rangle, \langle {\sf glc}, Y_{hg} \rangle \}).
{\sf HEX} $,

\noindent $\mathsf{HG} = \omega({\sf glc}, X_{hg}) . {\sf HEX} + \epsilon(t_h).\langle {\sf atp},
X_{ha} \rangle . \rho(\{ \langle {\sf atp}, X_{ha} \rangle, \langle {\sf glc}, Y_{hg} \rangle \}).
{\sf HEX}$,

\smallskip

\noindent and $X_{ha}, Y_{hg}$ are the surfaces of contact shown in Fig.~\ref{fig:glyco}. $ATP =
S_a[\mathsf{ATP}]$ models an ATP molecule where:

\smallskip

${\sf ATP}
= \langle \overline{{\sf atp}}, X_{ah} \rangle . ( \epsilon(t_a) . \rho(\{\langle \overline{{\sf
atp}}, X_{ah} \rangle\}) . {\sf ADP} + \omega(\overline{{\sf atp}}, X_{ah}) . {\sf ATP})$

\smallskip

\noindent and the surface of contact $X_{ah}$ is the whole boundary $\boundary{S_a}$.The process
modelling a molecule of glucose is similar: $GLC = S_g[\mathsf{GLC}]$ where

\smallskip

$\mathsf{GLC} = \langle
\overline{\mathsf{glc}}, X_{gh} \rangle . ( \epsilon(t_g) . \rho(\{\langle \overline{{\sf glc}},
X_{gh} \rangle\}) . \mathsf{G6P} + \omega(\overline{\mathsf{glc}}, X_{gh}) . \mathsf{GLC}$

\smallskip

\noindent We leave unspecified the behaviours $\mathsf{G6P}$ and $\mathsf{ADP}$.

${\sf HEX}$ has two channels $\langle {\sf atp}, X_{ha} \rangle$ and
$\langle {\sf glc}, Y_{hg}\rangle$ to bind, resp., with an ATP and a GLC molecule.
By performing an action $\langle \mathsf{atp}, X_{ha} \rangle$, ${\sf HEX}$ evolves in ${\sf HA}$.
${\sf HA}$ can perform either a weak-split action $\omega({\sf atp}, X_{ha})$ to come back to
${\sf HEX}$, or can wait at most $t_h$ units of time, perform $\langle \mathsf{glc}, Y_{hg}
\rangle$ and then evolve in $\rho(\{ \langle \mathsf{atp}, X_{ha} \rangle, \langle \mathsf{glc},
Y_{hg} \rangle \}). \mathsf{HEX}$. Now, two strong-split actions are enabled after which we come
back to ${\sf HEX}$. Notice that, after an action $\langle \mathsf{glc}, Y_{hg} \rangle$, {\sf HEX} becomes {\sf HG} that behaves symmetrically.

An ATP molecule performs a $\langle \overline{\mathsf{atp}}, X_{ah} \rangle$-action, waits $t_{r}$
units of time, and then can release the bond established on the channel $\langle \overline{{\sf
atp}}, X_{ah} \rangle$ -- and thus return free as ATP  -- or can participate in the reaction and become an ADP. As we will see in Section~\ref{sec:networks}, the result is the split of the complex in the three original shapes whose behaviours are {\sf HEX}, {\sf ADP} and {\sf G6P}, resp. We omit the description of the behaviour of ${\sf GLC}$ since it is similar to that of ${\sf ATP}$.
\end{example}

We are ready to define the timed operational semantics of 3D processes.

\begin{definition}[Transitional semantics of 3D processes]
\label{def:tbehprocesses}
Rules in Table~\ref{table:proc-behaviour} define the transition relations $\wnar{t} \subseteq (\proc
\times \proc)$ for $t \in \timedomain$, and $\nar{\mu} \subseteq \proc \times \proc$ for  $\mu
\in \Act$.
Two 3D processes $P$ and $Q$ are said to be {\em compatible}, written $P \comp Q$, if $P\nar{\langle
\alpha, X\rangle}$  and $Q\nar{\langle \overline{\alpha}, Y\rangle}$ for some compatible channels
$\langle \alpha, X\rangle$ and $\langle \overline{\alpha}, Y\rangle$;
otherwise, $P$ and $Q$ are {\em incompatible} that we denote with $P \not\comp Q$.
Below, we often write $P \not \nar{\rho}$ and $P \not\nar{\omega}$ as shorthand for
$P\not\nar{\rho(\alpha, X)}$ and $P\not\nar{\omega(\alpha, X)}$, resp., for any $\langle \alpha, X
\rangle$.
\end{definition}

Essentially, rules in Table~\ref{table:proc-behaviour} say that a 3D process inherits its functional and temporal behaviour from the $\bpa$-terms defining its internal behaviour. But now sites of binding capabilities and split actions are expressed w.r.t.\ a {\em global} coordinate system (see rules~\name{Basic$_c$} and~\name{Basic$_s$}). For simplicity, we have omitted a rule defining which weak-split action a basic process can perform. This can be obtained from rule~\name{Basic$_w$} by replacing each $\rho()$-action with a corresponding $\omega()$-action. It is worth noting that, due to rule \name{Comp$_{a2}$}, some of the $\langle \alpha, Y \rangle$-actions performed by $P$ (by $Q$) can be prevented in $P \unionc{a, X} Q$ since, due to binding a part (or all) of $Y$ has became interior because it is covered by a piece of $Q$ (of $P$ respectively) the surface of contact $Y \not \subseteq \boundary{P \unionc{a, X} Q}$ and, hence the corresponding channel is no more active. We have also omitted symmetric rules for~\name{Comp$_{a1}$} and~\name{Comp$_{a2}$} for the actions of $Q$.

The following proposition (see Appendix B for the proof) shows that 3D processes
well-formedness is closed w.r.t.\ transitions $\wnar{t}$ and $\nar{\mu}$.

\begin{proposition}\label{prop:closure-3d}
Let $P \in \proc$ well-formed. Either $P \wnar{t} Q$ or $P \nar{\mu} Q$ implies $Q \in \proc$
well-formed.
\end{proposition}

\begin{table}[tbp]
$$
\begin{array}{|c|}
\hline
\name{Basic$_t$}\sos{ B \wnar{t} B'}{
S[B] \wnar{t} (S+t)[B']} \\
\name{Comp$_t$} \sos{P \wnar{t} P' \quad Q \wnar{t} Q'  \quad X' = X + (t \cdot \velocity{P})}
{P \unionc{a, X} Q \wnar{t} P' \unionc{a, X'} Q'}\\ \\
\name{Basic$_{c}$}\sos{%
B \nar{\langle \alpha, X \rangle} B' \qquad Y  = \reference(X, \referencepoint{S})}
{ S[B] \nar{\langle \alpha, Y \rangle} S[B']} \\
\name{Basic$_{s}$}\sos{B \nar{\rho(\alpha, X)} B' \qquad Y = \reference(X, \referencepoint{S})}
{ S[B] \nar{\rho(\alpha, Y)} S[B']}\\
 \name{Comp$_s$} \sos{P \nar{\rho(\alpha, Y)} P'} {P
\unionc{a, X} Q \nar{\rho(\alpha, Y)} P' \unionc{a, X} Q} \\
\name{Comp$_w$} \sos{P \nar{\omega(\alpha, Y)} P'} {P
\unionc{a, X} Q \nar{\omega(\alpha, Y)} P' \unionc{a, X} Q} \\
\name{Comp$_{c}$} \sos{P \nar{\langle \alpha, Y \rangle } P' \qquad Y \subseteq \boundary{P
\unionc{a,X} Q}} {P \unionc{a, X} Q \nar{\langle \alpha, Y \rangle} P'
\unionc{a, X} Q}\\
\hline
\end{array}
$$
\caption{Functional and temporal behaviour of $\proc$-terms}
\label{table:proc-behaviour}
\end{table}

At this stage a key observation is that the operational rules in Table~\ref{table:proc-behaviour} do
not allow synchronization between components of compound process that proceed independently to each
other. Consider, as an example,
\begin{description}
\item $P = S_p [\rho(\{a, X_p\}). B_p]$,
\item $Q =S_q[\rho(\{\overline{a},X_q\}). B_q]$, and  $P\unionc{a,X} Q$
\end{description}
where $X = X'_p \cap X'_q$ and, for each $i \in \{p,q\}$, $X'_{i}$ is the site $X_i$ w.r.t.\ a global coordinate system, i.e. $X'_{i} = \reference(X_i, \referencepoint{S_i})$. As stand-alone processes, $P$ and $Q$ can perform two compatible strong-split actions, namely $\rho(a, X'_p)$ and $\rho(\overline{a},X'_q)$ and evolve, resp., in $S_p[B_p]$ and $S_q[B_q]$. As a consequence, $P \unionc{a,X} Q$ becomes either $S_p[B_p] \unionc{a,X} Q$ or $P \unionc{a,X} S_p[B_q]$.

But these actions are compatible and, according to the intuition given so far, $P$ and $Q$ have to synchronize on their execution in order to split the bond $\ch{a}{X}$. In other terms, a strong-split operation is enabled; such an operation must be performed before time can pass further and must produce as a result two independent 3D processes, i.e.\ the {\em network} of 3D processes (see Section~\ref{sec:networks}) that contains both $S_p[B_p]$ and $S_q[B_q]$. Similarly, we would allow synchronizations between compatible weak-split action in order to perform a weak-split operation. To properly deal with this kind of behaviours some technical details are still needed. We first allow synchronization on compatible split actions by introducing the transition relations  $\dnar{\rho(a, X)}$ and $\dnar{\omega(a,X)}$. Intuitively, we want that  $P \unionc{a, X} Q \dnar{\rho(a, X)} S_p[B_p] \unionc{a, X} S_q [B_q]$. Now, we can `physically' remove the bond $\langle a,X \rangle$ (this will be done by exploiting the function $\Split$ over 3D processes we provide in the next section) and obtain the network of processes we are interested in.

\begin{definition}[Semantics of strong and weak splits] \label{def:splittings}
The SOS-rules \\ \noindent that define the transition relations $\dnar{\rho(a, X)} \subseteq \proc
\times \proc$ where $\rho(\alpha,X) \in \rho(\channels)$ are given in Table~\ref{table:react-rules}.
As usual, symmetric rules have been omitted. We also omit the rules defining the transition
relations $\dnar{\omega(a,X)} \subseteq \proc \times \proc$ for $\omega(a, X) \in
\omega(\channels)$ since these  can be obtained from those in Table~\ref{table:react-rules}
by replacing each action $\rho(\mbox{-})$ with the corresponding action $\omega(\mbox{-})$.
\end{definition}

\begin{table}[bh]
$$
\begin{array}{|c|}
\hline
\name{StrSync} \sos{%
P \nar{\rho(\alpha, X_p) } P' \quad
Q \nar{\rho(\overline{\alpha}, X_q) } Q' \quad
\alpha \in \{a, \overline{a}\} \quad
X = X_p \cap X_q
}
{P \unionc{a, X} Q \dnar{\rho(a, X)} P' \unionc{a, X} Q'}\\
\name{StrPar} \sos{P \dnar{\rho(b, Y)} P'}
{P \unionc{a, X} Q \dnar{\rho(b, Y)} P' \unionc{a, X} Q} \\
\hline
\end{array}
$$
\caption{Transitional semantics for strong-split actions}
\label{table:react-rules}
\end{table}

Recall that strong-split operations require simultaneous split of multiple bonds. In
this case, all the components involved in the reaction must {\em all together} be ready to
synchronize on a proper set of compatible strong-split actions. Consider a more complex example
\begin{description}
\item $P = S_p[\rho(\{\ch{a}{X_p}, \ch{\overline{b}}{Y_p} \}). B_p]$,
\item $Q = S_q[\rho(\{\ch{\overline{a}} {X_q}\}).B_q]$,
\item $R = S_r[\rho(\{\ch{b}{Y_r} \}). B_r]$, and $(P\unionc{a,X} Q) \unionc{b,Y} R$
\end{description}

where $X = X'_p \cap X'_q$ and $Y = Y'_p \cap Y'_r$ (also in this case we write $X'_i$, for $i \in
\{p, q,r\}$, to represent the site $X_i$ in a global coordinate system).  $P$ can
synchronize with $Q$ and $R$ to split, resp., the bonds $\langle a, X \rangle$ and $\langle b, Y
\rangle$. Indeed, rules in Table~\ref{table:react-rules} implies that
$$
\begin{array}{l c l}
(P\unionc{a,X} Q) \unionc{b,Y} R & \dnar{\rho(a, X)} &
(S_p[\rho(\{\ch{\overline{b}}{Y_p} \}). B_p]\unionc{a,X} S_q[B_q]) \unionc{b,Y} R \\
& \dnar{\rho(b,Y)} & (S_p[B_p]\unionc{a,X} S_q[B_q]) \unionc{b,Y} S_r[B_r]
\end{array}
$$
After that, all `pending strong-split requests' of $(P \unionc{a,X} Q) \unionc{b,Y} R$ are
satisfied. We say that such a compound process is {\em able to complete a reaction}. If it was, for instance, $R = S_r[\epsilon(t).\rho(\{\ch{\overline{b}}{Y_r} \}). B_r]$ then $(P \unionc{a,X} Q) \unionc{b,Y} R$ would {\em not} have been able to complete a reaction, since (at least) one involved component, i.e.\ $R$, would not have been able to contribute to the reaction before than $t$ units of time. In such a case, the bonds can not be spit at all. This concept is formalized by the definition below.

\begin{definition}[Bonds of $\proc$-terms]\label{def:bonds}
The function $\bonds : \proc \rightarrow \wp(\channels)$ returns the set of bonds that are currently
established in $P$. It can be defined by induction on $P\in \proc$:
\begin{tabbing}
Basic: \quad \= $\bonds(S[B]) = \emptyset$\\
Comp: \> $\bonds(P \unionc{a,X} Q) = \bonds(P) \cup
\bonds(Q) \cup \{\langle a, X \rangle \}$
\end{tabbing}

By an easy induction on $P$ we can prove that $P \dnar{\rho(a, X)} $ implies $\langle a, X \rangle
\in \bonds(P)$. Moreover, we say that $P \in \proc$ is {\em able to complete a reaction}, which we
write as $P\searrow$, iff either (1) $P\not\nar{\rho}$, or (2) $P \dnar{\rho(a, X)} Q$ for some
$\rho(a, X)$ and $Q$ such that $Q \searrow$.
\end{definition}

Finally, if $P$ is able to complete a reaction and there is at a least a bond that has to be
strongly split (i.e.\ if $P \not\nar{\rho}$ does not hold), a reaction can actually take place
and, as a consequence, time cannot pass further. Below we restrict the timed operational
semantics of 3D processes as it has been defined in Def.~\ref{def:tbehprocesses} in order to take
this aspect into account.

\begin{definition}\label{def:strong-timed-behaviour}
Let $P\in\proc$. We say that $P \nar{t} Q$ if $P \wnar{t} Q$ and either $P \not\nar{\rho}$ or $P$
is {\em not} able to complete a reaction.
\end{definition}

\noindent Proposition~\ref{prop:splitwf} states that the function $\Split$ is well-defined up to
structural congruence over 3D processes we define below.

\begin{definition}[Structural congruence over 3D processes] \label{def:congruence3d}
We defi-\\ne the structural congruence over processes in $\proc$, which we denote by
$\equiv_P$, as the smallest relation that satisfies the following axioms:
\begin{description}

\item [-] $S[B] \equiv_P S'[B]$  provided that $S \equiv_S S'$;

\item [-] $P \unionc{a,X} Q  \equiv_P Q \unionc{a,X} P$;

\item [-] $P \equiv_P Q$ implies $P \unionc{a,X} R  \equiv_P Q \unionc{a,X} R$;

\item  [-] $Y \subseteq \boundary{Q} \cap \boundary{R}$ implies $(P \unionc{a,X} Q) \unionc{b, Y} R \equiv_P P \unionc{a,X} (Q \unionc{b, Y} R)$.
\end{description}
\end{definition}
\begin{proposition}\label{prop:splitwf}
Let $P \in \proc$ well-formed. If $\langle a, X \rangle \in \bonds(P)$ there is a well-formed 3D
process $Q \unionc{a, X} R \equiv_P P$.
\end{proposition}
We also need the following closure result.
\begin{proposition}\label{prop:closure-splittings}
Let $P, Q\in \proc$  and $\mu \in \omega(\channels) \cup \rho(\channels)$. Then:

\noindent
1. $P \dnar{\mu} Q$ implies $\shape(Q) = \shape(P)$;

\noindent
2. $P$ well-formed and $P\dnar{\mu} Q$ implies $Q$ well-formed.
\end{proposition}

\section{Networks or 3D processes}\label{sec:networks}

Now we can define a network of 3D processes as a collection of 3D processes moving in
the same 3D space.

\begin{definition}[Networks of 3D processes]
\label{def:networks}
The set $\nets$ of {\em networks of 3D processes} (3D networks, for short) is generated by the
grammar:
$$N ::= \NIL \gor P \gor N\,\|\, N $$
where $ P \in \proc$. Given a finite set of indexes $I$, we often write $ (\| \, P_i)_{i \in I}$
to denote the network that consists of all $P_i$ with $i \in I$. We assume that $I = \emptyset$
implies $(\|\, P_i)_{i \in I} = \NIL$. For $N  =(\| \, P_i)_{i \in I}$ we let $S_i = \shape(P_i)$,
for $i \in I$, and define $\colliding(N)$ as the set of all tuples $\langle P_i, P_j, X \rangle$  such that $\langle S_i, S_j, X \rangle \in \colliding(\{S_i\}_{i \in I})$ (see Def.~\ref{def:colls-information}).
A network $N = (\| \, P_i)_{i \in I}$ is said to be well-formed iff each $P_i$ is well-formed and,
for each pair of distinct processes $P_i$ and $P_j$, the shapes $S_i$ and $S_j$ {\em do not
interpenetrate}. Moreover, we extend the definition of $\move$ on networks in the natural way,
i.e.\ $\move \; t \;  (\| \, P_i)_{i \in I} = (\| \, \move \; t \ P_i)_{i \in I}$.
\end{definition}

In our running example we construct a network of processes containing a proper number of HEX, ATP
and GLC processes in order to replicate the conditions in a portion of cytoplasm.

\begin{definition}[Splitting bonds]\label{def:split}
The function $\Split: \proc  \times \wp(\channels) \rightarrow \nets$ is defined as follows:
If $\langle a, X \rangle \in \bonds(P) \cap C$ and $P \equiv_P Q \unionc{a, X} R$ then $\Split(P,
C) = \Split(Q, C) \,\|\, \Split(R, C)$; if, otherwise, $\bonds(P)  \cap C = \emptyset$, then
$\Split(P, C) = P$.
\end{definition}

It is worth noting that split shapes maintain the same velocity until the next occurrence of a
movement time step. As we mentioned above, this is not a problem because they will not trigger a
collision and, thus, a shorter first time of contact.

\begin{proposition}\label{prop:split-closure}
Let $P \in \proc$ well-formed and $C \subseteq \channels$. Then $\Split(P,C)$ is a well-formed
network of 3D processes.
\end{proposition}

\begin{definition}[Semantics of weak- and strong-split operation] \label{def:creaction}
Let\\
\noindent $P \in \proc$ a 3D process. If $P\searrow$, we write that $P \nar{\rho} N \in \nets$ iff
there is a non empty set of channels $C =\{\langle a_1, X_1 \rangle, \cdots, \langle a_n,X_n \rangle
\} \subseteq \bonds(P)$ such that $P=P_0 \dnar{\rho(a_1, X_1)} P_1 \cdots \dnar{\rho(a_n, X_n)}
P_n$, $P_n \not\nar{\rho}$  and $N = \Split(P_n,C)$. Similarly, $P \nar{\omega} N$ iff there is a
channel $\langle a, X \rangle \in \bonds(P)$ such that $P \dnar{\omega(a, X)} Q$ and $N = \Split(Q,
\{\langle a, X \rangle\})$.
\end{definition}

Since weak-split operations are due to a synchronization between just a pair of 3D processes,
condition `$P$ is able to complete a reaction' is not needed, but `$P \dnar{\omega(a, X)} Q$'
suffices to our aim.

\begin{example}
Let us consider $P = H \unionc{{\sf atp}, X} (A  \unionc{{\sf glc}, Y}\,G)$ where:

\begin{description}
\item $H = S_{h}[\,\rho(\{\ch{{\sf atp}}{X_{ha}}, \ch{{\sf glc}}{Y_{hg}}\}). {\sf HEX}\,]$,

\item  $A=S_{a}[\, \rho(\{ \langle \overline{{\sf atp}}, X_{ah} \rangle \}). {\sf ADP}
+ \omega(\overline{{\sf atp}}, X_{ah}). {\sf ATP} \, ]$,

\item $G=S_{g}[\rho(\{ \langle \overline{{\sf glc}}, Y_{gh} \rangle\}). {\sf G6P} +
\omega(\overline{{\sf glc}}, Y_{gh}). {\sf GLC}]$,
\end{description}

\noindent $X'_{ha} \cap X'_{ah} = X$ (here  $X'_{ha}$ and $X'_{ah}$ are the sites $X_{ha}$
and $X_{ah}$ expressed w.r.t.\ a global coordinate system; this convention will be applied later
on)  and $Y'_{hg} \cap Y'_{gh} = Y$. According to the definitions given so far, $P$ is {\em able to complete a reaction} since:
$$\begin{array}{l c l}
P & \dnar{\rho({\sf atp}, X)} &
S_{h}[\rho(\{\langle {\sf glc}, Y_{hg} \rangle\}). {\sf HEX}] \unionc{{\sf atp}, X}
(S_{a}[{\sf ADP}]  \unionc{{\sf glc}, Y} \, G ) \\
& \dnar{\rho({\sf glc}, Y)} &
S_{h}[ {\sf HEX} ] \unionc{{\sf atp}, X} (S_{a}[ {\sf ADP}]  \unionc{{\sf
glc}, Y} S_g [ {\sf G6P} ]) = R
\end{array}$$
\noindent Moreover, $R  \not\nar{\rho}$ and
$$\begin{array}{l c l}
\Split(R, C)  & = &
S_a [{\sf HEX}] \,\|\, \Split(S_{a}[{\sf ADP}] \unionc{{\sf glc}, Y} S_{g}[{\sf G6P}],C) \\ & = &
S_{h}[{\sf HEX}] \,\|\, ( S_{a}[{\sf ADP}] \,\|\, S_{g}[{\sf G6P}] ) = N
\end{array}$$
\noindent  where $C = \{\langle {\sf atp}, X \rangle, \langle {\sf glc}, Y\rangle \} $, implies $P
\nar{\rho} N$. Moreover, let us note that, for each $t \in \timedomain$, $P \wnar{t} P$ but $P
\not\nar{t}$ since since $P$ is able to complete a reaction and $P \not\nar{\rho}$ does not hold.
\end{example}

Below we define the temporal and functional behaviour of 3D networks. We assume that
such networks perform basic actions that belong to set $\{\omega, \rho, \kappa\}$, where $\omega$ and $\rho$ denote, resp., weak- and strong- split operations as a unique action (at the network level we only see whether shape bonds can be split or not) and $\kappa$ represents system evolutions due to collision response (see Section~\ref{sec:response}). We also let elements of the set $\{\omega, \rho\} \cup \timedomain$ to be ranged over by $\nu, \nu', \cdots$.

\begin{definition}[Temporal and Functional Behaviour of $\nets$-terms]\label{def:nets-behaviour}
\ \\ Rules in Table~\ref{table:nets} defines the transition relations $\nar{t} \subseteq \nets \times \nets$  for $t \in \timedomain$ and $\nar{\nu} \subseteq \nets \times \nets$ for $\nu \in
\{\omega,\rho\}$. As usual, symmetric rules have been omitted.

A {\em timed trace} from  $N$ is a finite sequence of steps of the form $N=N_0 \nar{\nu_1} N_1
\cdots \nar{\nu_{n}} N_{n} = M$. We also write that $N \dnar{t} M$ if there is a timed trace
$N=N_0 \nar{\nu_1} N_1  \cdots \nar{\nu_{n}} N_{n} = M$ such that $t = \sum\limits_{i=0}^{n} \{\nu_i \,|\, \nu_i \in \timedomain \} $.
\end{definition}

\begin{table}[th]
\[
\begin{array}{|c|}
\hline
\name{Empty$_{t}$}\sos{}{
\NIL \nar{t} \NIL} \quad
\name{Par$_t$} \sos{N \nar{t} N' \quad M \nar{t} M'}
{N \,\|\, M \nar{t} N'\,\|\, M'}  \\
\name{Par$_{a}$}\sos{N \nar{\nu} N'}{N \,\|\, M \nar{\nu} N'\,\|\, M}\\
\hline
\end{array}\]
\caption{Temporal and functional behaviour of 3D networks}
\label{table:nets}
\end{table}

\begin{proposition}\label{prop:net-closure}
Let $t \in \timedomain$, $P\in \proc$, $N\in \nets$, with $P$ and $N$ well-formed.

\par\smallskip\noindent 1. $P \nar{\omega} N$ and $P \nar{\rho} N$ implies $N \in \nets$
well-formed.

\par\smallskip\noindent 2. $N \nar{t} M$ implies $M$ well-formed;

\par\medskip\noindent 3. $N \dnar{t} M$ implies $M$ well-formed.
\end{proposition}

\subsection{Collision response}\label{sec:response}

In this section we describe the semantics of {\em collisions response}. As already said, the notion of compatibility between channels (and, hence, processes) has been introduced to distinguish between elastic and inelastic collision. In particular, collisions among compatible processes are always inelastic. So, if $P\nar{\ch{a}{X_p}} P'$ and $Q \nar{\ch{\overline{a}}{X_q}} Q'$, with $\langle a, X_p\rangle$ and $\langle \overline{a}, X_q\rangle$ compatible, and $P$ and $Q$ collide in the non-empty site $X = X_p \cap X_q$ we get a compound process $(P' \unionc{a, X} Q') \lsbrace \vect{v} \rsbrace $ where the velocity $\vect{v}$ is  provided by Equation~(3) in Def.~\ref{def:collisions}. Vice versa, a collision between two incompatible processes $P$ and $Q$ is treated as an elastic one. After such a collision, $P$ and $Q$ (actually the processes we obtain by updating their velocities according to Equations (1) and (2) in Def.~\ref{def:collisions}) will proceed independently to each other.

To resolve collisions, we introduce two different kinds of {\em reduction relations} over 3D
networks, namely $\enar{\langle P, Q, X \rangle}$ and $\inar{\langle P, Q, X \rangle}$, where $P, Q$ are 3D processes and $X$ is a surface of contact (see Table~\ref{table:reduce}). Intuitively, if $N \enar{\langle P, Q, X \rangle} M$ ($N \inar{\langle P, Q, X \rangle} M$), then $M$ is the 3D network we obtain once an elastic (inelastic, resp.) collision between $P$ and $Q$ in the surface of contact $X$ has been resolved. These reduction relations also use the structural congruence over 3D networks.

\begin{definition}[Structural congruence over 3D networks]\label{def:congruencenets}
\ \\ The structural congruence over terms in $\nets $, that we denote with $\equiv$, is the
smallest relation that satisfies the following axioms:

\begin{itemize}
\item [-] $\NIL \,\|\, N \equiv N$,  $N \,\|\, M \equiv M \,\|\, N$ and $N
\,\|\, (M \,\|\, R)
\equiv (N \,\|\, M)  \,\|\, R$;

\item [-] $P \,\|\, N \equiv Q \,\|\,N$ provided that $P \equiv_P Q$.
\end{itemize}
\end{definition}

Rule \name{elas} in Table~\ref{table:reduce} simply changes velocities of two colliding but
incompatible processes guided by Equations (1) and (2) in Def.~\ref{def:collisions}, while
rule~\name{inel} joins two compatible processes $P$ and $Q$ to obtain a compound process
whose velocity is given by Equation (3) in Def.~\ref{def:collisions}. Note that we force $P$ and $Q$
to synchronize on the execution of two compatible actions $\ch{\alpha}{X_p}$ and
$\ch{\overline{\alpha}}{X_q}$ before joining them. In rules~\name{elas$_\equiv$}
and~\name{inelas$_\equiv$} we also consider structural congruence over nets of processes.
In Def.~\ref{def:reductionrules} we collect together all the reduction-steps needed to
solve collisions listed in a given set of collisions $\colliding(N)$; clearly $N$ is a generic  3D
network.

\begin{definition}[Resolving  collisions]
\label{def:reductionrules}
Let $N \in \nets$ and $\langle P, Q, X \rangle$ a tuple in $\colliding(N)$. $N
\nar{\langle P,Q, X \rangle} M$ if either $P \comp Q$ and $N \inar{\langle P,Q, X \rangle} M$ or $P \not\comp Q$ and $N \enar{\langle P,Q, X \rangle} M$.

Moreover, we write that $N \nar{\kappa} M$ if either $\colliding(N) = \emptyset$ and $N=M$  or
$\colliding(N) \neq \emptyset$ and there is a finite sequence of reduction
steps $N = N_0 \nar{\langle P_1,Q_1, X_1 \rangle} N_1 \cdots \nar{\langle P_k,Q_k, X_k \rangle} N_k = M$ such that:
\begin{enumerate}
\item $\langle P_{i},Q_i, X_i \rangle \in \colliding(N_{i-1})$ for each $i \in [1, k]$;
\item $\colliding(N_k) = \emptyset$.
\end{enumerate}
\end{definition}

\begin{table}[tbh]
$$
\begin{array}{|c|}
\hline
\name{elas} \sos{P \collision{X}{e} Q = (\vect{v}_p,\vect{v}_q)}{%
(P \,\|\, Q) \,\|\, N \enar{\langle P,Q,X \rangle} (P \lsbrace \vect{v}_p\rsbrace \,\|\, Q\lsbrace
\vect{v}_q \rsbrace) \,\|\, N}\\
\name{elas$_\equiv$}\sos{%
N \equiv N' \quad N' \enar{\langle P, Q, X \rangle} M} { N \enar{\langle P, Q, X \rangle} M}\\ \\
\name{inel} \sos{%
P \nar{\langle \alpha, X_p\rangle} P' \quad Q \nar{\langle \overline{\alpha}, X_q \rangle} Q' \quad
\alpha \in \{a,\overline{a} \}  \quad P \collision{X_p \,\cap\, X_q}{i} Q = \vect{v}}{%
(P \,\|\, Q) \,\|\, N \inar{\langle P, Q, X_p \,\cap\, X_q \rangle}
((P' \unionc{a,X_p \cap X_q} Q') \lsbrace \vect{v} \rsbrace) \,\|\, N
}\\
\name{inel$_\equiv$} \; \sos{N \equiv N' \qquad
N' \inar{\langle P, Q, X \rangle} M}
{N \inar{\langle P, Q, X \rangle} M }\\
\hline
\end{array}$$
\caption{Reaction rules for elastic and inelastic collisions}
\label{table:reduce}
\end{table}

Let also note that, at any given time $t$, $\colliding(N)$ can be obtained from the set of all the pairs of processes in $N$ that are touching at that time. This set and hence $\colliding(N)$ is surely finite and changes only when we resolve some inelastic collision (this is because, after an inelastic collision one or more binding sites can possibly become internal points of a compound process, and hence are not available anymore). Moreover a collision between pairs of processes with the same shape can not be resolved twice. This is either because two processes $P$ and $Q$  have been bond in a compound process as a consequence of an inelastic collision, or because $P$ and $Q$ collide elastically and their velocities have been changed according to Equations (1) and (2) in Def.~\ref{def:collisions} in order to obtain two processes that do not collide anymore (see Lemma~\ref{lemma:reduction-steps} in Appendix C). Thus, we can always decide if there is a finite sequence of reduction steps that allows us to resolve all collisions listed in $\colliding(N)$  and hence obtain a network $M$ with $\colliding(M) = \emptyset$.

\begin{proposition}\label{prop:reduction-closure}
Let $N \in \nets$, $P,Q \in \proc$ and $X$ a not-emptyset subset of $\pos$. Then $N$ well-formed
and $N \nar{\langle P, Q, X \rangle} M$ implies $M$ well-formed.
\end{proposition}

By iterative applications of Proposition~\ref{prop:reduction-closure} (see Appendix C for the proof)
it is also:

\begin{lemma}\rm\label{lemma:reduction-closure}
Let $N$ a well-formed 3D network. Then $N \nar{\kappa} M$ implies $M$ well-formed.
\end{lemma}

We are now ready to define how a network of 3D processes evolves by performing an infinite number of {\em movement time steps}.

\begin{definition}[System evolution]
\label{def:mts}
Let $N, M \in \nets$  and $t, t'\in \timedomain$. We say that $(N, t) \dnar{t'} (M, t + t')$ iff one of the following conditions holds:

\begin{enumerate}
\item $t' =\ftoc (N) \leq \mts$ and $N \dnar{t'} N' \nar{\kappa} N'' $ and $M = \move \;  (t+ t') \; N''$;

\item $t' = \mts < \ftoc (N) $ and $N \dnar{t'} N'$ and $M = \move \; (t+ t') \; N'$.
\end{enumerate}

A {\em system evolution} is any infinite sequence of time steps of the form:
$$(N_0,0) \dnar{t_1} (N_1, t_1) \dnar{t_2} (N_2, t_1 + t_2)
\cdots (N_{i-1},  \sum\limits_{j=1}^{i-1} t_j) \dnar{t_i} (N_{i}, \sum\limits_{j=1}^{i} t_j )
\dnar{t_{i+1}} \cdots$$
\end{definition}

Note that, for each $i \geq 1$, $t_i = \min \{\ftoc(N_{i-1}), \mts\}$ as discussed in Section~\ref{sec:overview}. Moreover, in order to make sure that processes will never interpenetrate during a system evolution, if $\ftoc(N_{i-1}) \leq \mts$, we first resolve all the collisions that happen after time $t_i = \ftoc(N_{i-1})$ (by means of transition $\nar{\kappa}$) and then apply the changes suggested by the function $\move$ as described in Section~\ref{sec:trajectories}.

\begin{example}\label{ex:evolution}
This example shows a possible evolution of the 3D network $ N_0 = (HEX \,\|\, ATP) \,\|\, GLC$ where $HEX$, $ATP$ and $GLC$ are the 3D processes of Example~\ref{ex:running}. Below we use the following notation:

\begin{itemize}
\item [ - ]  $H(t)= (S_h + t) [{\sf HEX}]$, $A(t) = (S_a + t) [{\sf ATP}]$ and $G(t) = (S_g + t)
[{\sf GLC}]$ for each $t \in \timedomain$. Note that $HEX=H(0)$, $ATP=A(0)$ and $GLC=G(0)$;

\item [ - ]  ${\sf C}=  \rho(\{\langle {\sf atp}, X_{ha} \rangle, \langle  {\sf glc}, X_{hg}
\rangle\}). {\sf HEX}$  and, for any $t \leq t_h$,  ${\sf HA}(t) = \omega({\sf atp}, X_{ha}). {\sf
HEX} + \epsilon(t_h -t). \langle  {\sf glc}, X_{hg}  \rangle.{\sf C}$;

\item [ - ] ${\sf AH}(t) = \omega(\overline{{\sf atp}}, X_{ah}). {\sf ATP} + \epsilon(t_a -t).
\rho(\{ \overline{{\sf atp}}, X_{ah} \}). {\sf ADP}$ for any $t \in \timedomain$ with $t \leq t_a$;

\item [ - ]  ${\sf GH}(t) = \omega(\overline{{\sf glc}}, X_{gh}). {\sf GLC} + \epsilon(t_g -t).
\rho(\{ \overline{{\sf glc}}, X_{gh} \}). {\sf G6P}$ for any $t \leq t_g$.
\end{itemize}

\noindent Let $t_1 = \ftoc(N_0)$ and assume $t_1 \leq \mts$. By the operational rules, it is $N_0
\dnar{t_1} H(t_1) \,\|\, A(t_1) \,\|\, G(t_1) = N'_0$. We also assume that $\colliding(N'_0) =
\{\langle H(t_1), A(t_1), X \rangle \}$ where $X = X'_{ha} \cap X'_{ah} \neq \emptyset$, $X'_{ha}
= \reference(X_{ha}, \referencepoint{S_h + t_1})$ and $X'_{ah} = \reference(X_{ah},
\referencepoint{S_a + t_1})$. Then: $N'_0 \nar{\kappa} P(t_1) \,\|\, G(t_1) = N''_0$ where $P(t_1)
= \big((S_h+t_1) [{\sf HA}(0)] \langle {\sf atp}, X \rangle (S_a+t_1)[{\sf AH}(0)] \big)\lsbrace
{\bf v}_{ha} \rsbrace$ and ${\bf v}_{ha}  =  H(t_1)  \collision{X}{i} A(t_1)$. Finally: $$(N_0, 0)
\dnar{t_1} (N_1, t_1)$$ where $N_1 = \move \, t_1 \, N''_0 =  \move \, t_1 \, P(t_1) \,\|\, \move
\, t_1 \, G(t_1) = \\ P(t_1) \lsbrace {\bf v}_{1} \rsbrace \,\|\, G(t_1) \lsbrace {\bf v}_{2}
\rsbrace$. Note that:

$$
\begin{array}{lcl}
P(t_1) \lsbrace {\bf v}_{1} \rsbrace  & =  & \big((S_h+t_1) [{\sf HA}(0)] \langle {\sf atp}, X
\rangle (S_a+t_1)[{\sf AH}(0)] \big)\lsbrace {\bf v}_{1} \rsbrace  \\
& = & \big( S_h^1 [{\sf HA}(0)] \langle {\sf atp}, X \rangle S_a^1 [{\sf AH}(0)] \big)  \\
\end{array}
$$

\noindent where $S_h^1 = ((S_h+t_1)\lsbrace {\bf v}_{1} \rsbrace)$ and $S_a^1 = ((S_a+t_1)\lsbrace
{\bf v}_{1} \rsbrace)$. Moreover, $G(t_1) \lsbrace {\bf v}_{2} \rsbrace  =  ((S_g + t_1)\lsbrace
{\bf v}_{2} \rsbrace [{\sf GLC}]= S_g^1 [{\sf GLC}]$.

Let $t_2 = \ftoc(N_1)$ and assume $t_2 = t_h \leq \min \{t_a, \mts\}$\footnote{If were $t_2 < t_h$. the 3D processes ${\sf HA}(t_2)$ and ${\sf GLC}$ would be no more compatible, and a collision between them would be treated as elastic. On the other hand, if were $t_2 = t_a$ the idling time for ${\sf AH}(t_a)$ would be over. As a consequence, time would pass further only after the execution of a weak-split operation that splits the bond between the Hexokinases and the Atp molecules}. Below we write $G'(t_2)$ and $P'(t_2)$ to denote, respectively, the 3D processes $(S_g^1 + t_2)[{\sf GLC}]$ and $\big( (S_h^1 + t_2) [{\sf HA}(t_h)] \langle {\sf atp}, X + t_2 \cdot {\bf v}_1 \rangle (S_a^1 + t_2) [{\sf AH}(t_2)] \big)$. Again by the operational rules, $N_1 \dnar{t_2}  P'(t_2) \,\|\, G'(t_2) = N'_1$.

Let $\colliding(N'_1) = \{\langle P'(t_2), G'(t_2), Y \rangle \}$ where $Y =
X'_{hg} \cap X'_{gh} \neq \emptyset$, $X'_{hg} = \reference(X_{hg},
\referencepoint{S^1_h + t_2})$ and $X'_{gh} =
\reference(X_{gh},\referencepoint{S^1_g + t_2}) \subseteq \\
\boundary{P'(t_2)}$. If ${\bf v}_{gh} = P'(t_2) \collision{Y}{i} G'(t_2)$, then
\\ $N'_1 \nar{\kappa} \big( P(t_2) \, \langle {\sf glc},Y \rangle \, G(t_2)\big)
\lsbrace {\bf v}_{gh} \rsbrace = N''_1$ where $G(t_2) = (S_g^1 + t_2)[{\sf
GH}(0)]$ and $P(t_2) = \big( (S_h^1 + t_2) [{\sf C}] \langle {\sf atp}, X + t_2
\cdot {\bf v}_1 \rangle (S_a^1 + t_2) [{\sf AH}(t_2)] \big)$.  Finally: $$(N_1,
t_1) \dnar{t_2} (N_2, t_1 + t_2)$$ where $N_2 = \move \; (t_1 + t_2) \;  N''_1
= \big( P(t_2) \, \langle {\sf glc}, Y \rangle \, G(t_2)\big) \lsbrace {\bf
v}_{3} \rsbrace$. Observe that:

$\big( P(t_2) \, \langle {\sf glc}, Y \rangle \, G(t_2)\big) \lsbrace {\bf v}_{3} \rsbrace =
\big( P(t_2)  \lsbrace {\bf v}_{3} \rsbrace \big) \, \langle {\sf glc}, Y \rangle \, \big(G(t_2)
\lsbrace {\bf v}_{3} \rsbrace \big) =$

$\big( (S_h^2 [{\sf C}] \, \langle {\sf atp}, X + t_2 \cdot {\bf v}_1 \rangle \, S_a^2 [{\sf
AH}(t_2)] \big) \, \langle {\sf glc}, Y \rangle \, S_g^2[{\sf GH} (0)]$

\noindent where $S_h^2 = (S_h^1 + t_2) \lsbrace {\bf v}_{3} \rsbrace$, $S_a^2 = (S_a^1 + t_2)
\lsbrace {\bf v}_{3} \rsbrace$ and $S_g^2 = (S_g^1 + t_2) \lsbrace {\bf v}_{3} \rsbrace$.

At this stage the network contains just one process and, as a consequence, no collisions are possible. Thus, $\ftoc(N_2) = \infty$. Assume $t_g = t_a - t_2 \leq \mts$\footnote{If were $t_g \neq t_a - t_2$ the reaction could never proceed since the involved molecules would never be able to release -- all together -- the bonds. Thus the system would deadlock.}. If we let $X_g = (X + t_2 \cdot {\bf v}_1) + t_g \cdot {\bf v}_3$ and $Y_g = Y + t_g \cdot {\bf v}_3$, then:

\noindent $N_2 \dnar{t_g} ((S_h^2 + t_g) [{\sf C}] \, \langle {\sf atp}, X_g \rangle \, (S_a^2  + t_g) [{\sf AH}(t_a)] \big) \, \langle {\sf glc}, Y_g \rangle (S_g^2 + t_g )[{\sf GH} (t_g)] \nar{\rho} ((S_h^2 + t_g) [{\sf HEX}] \,\|\, (S_a^2  + t_g) [{\sf ADP}] \big) \, \|\, (S_g^2 + t_g )[{\sf G6P}] \dnar{\mts - t_g} \\ ((S_h^2 + \mts) [{\sf HEX}] \,\|\, (S_a^2  + \mts) [{\sf ADP}] \big) \, \|\, (S_g^2 + \mts )[{\sf G6P}] = N'_2$. Thus:

$$
(N_2, t_1 + t_2) \dnar{\mts} (N_3, t_1 + t_2 + \mts)
$$

where $N_3 = \move \; (t_1 + t_2 + \mts) \; N'_2 = (S_h^3 [{\sf HEX}] \,\|\, S_a^3 [{\sf ADP}] \big) \, \|\, S_g^3 [{\sf G6P}]$, $S_i^3 = (S_i^2 + \mts) \lsbrace {\bf v}_i \rsbrace $ and ${\bf v}_i = \move \; (t_1 + t_2 + \mts) \; (S_i^2 + \mts)$ for each $i \in \{h, a, g\}$.
\end{example}

\medskip

We can prove the following basic property of the Shape Calculus stating that any system evolution
does not introduce space inconsistencies like interpenetration of 3D processes or not well-formed
processes.

\begin{theorem}[Closure w.r.t.\ well-formedness]
Let $N$ be a well-formed network of 3D processes. If $(N,t) \dnar{t'} (M,t+t')$ then $M$ is
well-formed.
\end{theorem}

\begin{proof}
Assume that $(N,t) \dnar{t'} (M,t+t')$ because of  $N \dnar{t'} N' \nar{\kappa} N''$ and $M = \move \; (t+t') \; N''$ (the other case  -- see Def.~\ref{def:mts} -- is similar). Then, by
Proposition~\ref{prop:net-closure}~(item~3) and Lemma~\ref{lemma:reduction-closure},  $N$ well-formed and
$N \dnar{t'} N' \nar{\kappa} N''$ implies $N''$  and hence $M = \move \; (t+t') \; N''$ well-formed.
\end{proof}

\section{Conclusions and Future Work}
\label{sec:conclusion}

We have defined the full timed operational semantics of the Shape Calculus and we have introduced a notion of well-formedness of the different objects of the calculus. We proved that the evolution of a well-formed network of 3D processes is always well-formed, that is to say, no spatial or temporal inconsistencies can be introduced by the dynamics of the calculus. As future work we intend to provide verification techniques for the Shape Calculus. In order to do this we believe that a sort of ``tailoring'' should be made on the calculus, making some parts (e.g.\ movement) more abstract and other parts (e.g.\ behaviours) more specific adding quantitative information (for instance probabilities or costs). The whole process will then be supported by the definition of proper logical languages to specify properties of interests. Of course we expect that some approximations will be necessary to reach computability and/or feasibility. Another direction of future work is the possibility to include in the calculus some new useful, but in some cases complex, concepts such as re-shaping, message passing of values, and communication by perception of a compatible process in the neighbourhood.

\bibliographystyle{abbrv}
\bibliography{shapecalculus2}

\newpage

\section*{Appendix A: Proofs of Section~\ref{sec:3ds}}
The following lemmas are first needed.
\begin{lemma}\label{lemma:shape-properties1}
Let $S \in \shapes$ well-formed and $t \in \timedomain$. Then:

1. $\velocity{S+t} = \velocity{S}$ and $\mass{S+t} = \mass{S}$;

2. $\referencepoint{S+t} = \referencepoint{S} + t \cdot \velocity{S}$;

3. $\points{S+t} = \points{S} + t \cdot \velocity{S}$ (and, hence, $\boundary{S+t} = \boundary{S}
+ t \cdot \velocity{S}$).
\end{lemma}
\begin{proof}
We prove Items 1, 2 and 3 by induction on $S \in \shapes$ well-formed.

\smallskip \noindent 
{\em Basic:} $S= \langle V, m, \vect{p}, \vect{v} \rangle $, $S+t = \langle V + (t \cdot \vect{v}),
m, \vect{p} + (t \cdot \vect{v}), \vect{v} \rangle$ and:
 
\noindent 1. $\velocity{S+t} = \vect{v} = \velocity{S}$ and $\mass{S+t} = m = \mass{S}$;

\noindent 2. $\referencepoint{S+t} = \vect{p} + (t \cdot \vect{v}) = \referencepoint{S} + (t \cdot
\vect{v})$;

\noindent 3. $\points{S+t} = V + (t \cdot \velocity{S})  = \points{S} + (t \cdot \velocity{S})$.

\par\medskip\noindent
{\em Comp:} $S= S_1 \union{X} S_2$ and $S+t = (S_1 +t) \union{Y} (S_2+t)$ where $Y = X + t
\cdot {\bf v}$ and ${\bf v} = \velocity{S} = \velocity{S_i}$. By induction hypothesis:

\par\medskip\noindent 1. $\velocity{S_i+t}= \velocity{S_i}= {\bf v}$ and $\mass{S_i+t} = \mass{S_i}$
for $i=1,2$.
Thus: $\velocity{S+t}= {\bf v}$ and $\mass{S+t} = \mass{S_1+t} + \mass{S_2+t} = \mass{S_i} +
\mass{S_2} = \mass{S}$.

\par\medskip\noindent 2.  $\referencepoint{S+t} = \frac{\sum\limits_{i=1}^2 \mass{S_i + t} \cdot
\referencepoint{S_i + t}} {\sum\limits_{i=1}^2 \mass{S_i + t}}
= \frac{\sum\limits_{i=1}^2 \mass{S_i} \cdot \big(\referencepoint{S_i } + t \cdot \vect{v} \big)}
{\sum\limits_{i=1}^2 \mass{S_i}}
= \frac{\sum\limits_{i=1}^2 \mass{S_i} \cdot \referencepoint{S_i } } {\sum\limits_{i=1}^2
\mass{S_i}} +  \\ t \cdot \vect{v} = \referencepoint{S} +  t \cdot \vect{v}$

\vspace*{0.2cm}

\noindent 3. $\points{S+t} = \points{S_1 +t} \cup \points{S_2 +t} = \points{S_1} \cup
\points{S_2}=\points{S}$.

\end{proof}

\par\medskip\noindent
To prove Prop.~\ref{prop:shapes-wellformedness} we also need the following lemma; its proof has been
omitted because it is similar to that of Lemma~\ref{lemma:shape-properties1}.
\begin{lemma}\label{lemma:shape-properties2}
Let $S \in \shapes$ well-formed and ${\bf w} \in \vel$. Then:

1. $\velocity{S\lsbrace\vect{w}\rsbrace} =\{\vect{w}\}$ and $\mass{S\lsbrace\vect{w}\rsbrace} =
\mass{S}$;

2. $\points{S\lsbrace\vect{w}\rsbrace} = \points{S} $ (and, hence,
$\boundary{S\lsbrace\vect{w}\rsbrace} = \boundary{S}$).
\end{lemma}
\par\medskip\noindent
{\bf Proposition~\ref{prop:shapes-wellformedness}}
{\em Let $S \in \shapes$, $t\in \timedomain$ and $\vect{w}\in\vel$. If $S$ is well-formed then $S+t$
and $S\lsbrace\vect{w}\rsbrace$ are well-formed.}

\begin{proof}
We prove these statement by induction in $S \in \shapes$.
\par\medskip\noindent
{\em Basic:} $S = \bs{V}{m}{\mathbf{p}}{v}$. Both
$S + t = \bs{V+ t \cdot {\bf v}}{m}{{\bf p}+t\cdot \mathbf{v}}{v}$ and
$S\lsbrace\vect{w}\rsbrace = \bs{V}{m}{\mathbf{p}}{w}$ are well-formed shapes.
\par\medskip\noindent
{\em Comp:} $S = S_1 \union{X} S_2$.   Let $\vect{v} = \velocity{S} = \velocity{S_i}$ for
$i=1,2$. By Defs~\ref{def:shapemoving} and~\ref{def:shape-velocity}, it is $S+ t = (S_1+t) \union{Y}
(S_2 + t)$ where $Y = X + t \cdot \vect{v}$, and $S\lsbrace\vect{w}\rsbrace =
(S_1\lsbrace\vect{w}\rsbrace) \union{X} (S_2\lsbrace\vect{w}\rsbrace)$.
We first prove $S +t$ wellformedness.

\par\medskip\noindent 1. By induction hypothesis, $S_1 + t$ and $S_2 +t$ are well-formed;

\par\medskip\noindent 2. $S$ well-formed implies that $X$ is non-empty set that is equal to
$\points{S_1} \cap \points{S_2} = \boundary{S_1} \cap \boundary{S_2}$. By
Lemma~\ref{lemma:shape-properties1}, we also
have that $Y = X + t \cdot \vect{v}$ is non-empty and $Y = \points{S_1 + t} \cap
\points{S_2 + t} = \boundary{S_1 + t} \cap \boundary{S_2 + t}$.

\noindent 3. By Lemma~\ref{lemma:shape-properties1}-1, $\velocity{S + t} = \vect{v}$ where
$\vect{v} = \velocity{S_i + t} = \velocity{S_i}$ for $i =1,2$.

\par\medskip\noindent Similarly we can prove that $S\lsbrace\vect{w}\rsbrace$ is well-formed. Shapes
$S_1\lsbrace\vect{w}\rsbrace$ and $S_2\lsbrace\vect{w}\rsbrace$ are well-formed (this follows by
induction hypothesis); by Lemma~\ref{lemma:shape-properties2}-1 it is
$\velocity{S\lsbrace\vect{w}\rsbrace} = \vect{w}$; finally $\emptyset \neq X = \points{S_1} \cap
\points{S_2} = \boundary{S_1} \cap \boundary{S_2}$,
$\points{S_i} = \points{S_i\lsbrace\vect{w}\rsbrace}$ and  $\boundary{S_i} = \boundary{S_i
\lsbrace\vect{w}\rsbrace}$ for $i=1,2$ (by Lemma~\ref{lemma:shape-properties2}-2) imply
$\emptyset \neq X = \points{S_1 \lsbrace\vect{w}\rsbrace } \cap \points{S_2
\lsbrace\vect{w}\rsbrace} = \boundary{S_1 \lsbrace\vect{w}\rsbrace} \cap \boundary{S_2
\lsbrace\vect{w}\rsbrace}$.
\end{proof}
\section*{Appendix B: Proofs of Section~\ref{sec:three}}
This section is devoted to prove main results stated in Section~\ref{sec:three}.

\par\medskip\noindent {\bf Proposition~\ref{prop:closure-3d}}
\em Let $P \in \proc$ well-formed. Either $P \wnar{t} Q$ or $P \nar{\mu}{} Q$ implies $Q \in \proc$
well-formed.\rm

\smallskip\noindent\begin{proof}
We first prove (by induction on $P\in \proc$) that:

\par\medskip\noindent 1. $P \wnar{t} Q$ implies $\shape(Q) = \shape(P) + t $;

\par\medskip\noindent 2. $P \nar{\mu} Q$ implies $\shape(Q) = \shape(P)$.

\par\medskip\noindent
{\em Basic:} $ P = S[B] $.

\noindent 1. By operational rules, $P \wnar{t} Q$ implies $B \nar{t} B'$ and $Q= (S+t)[B']$. Thus:
$\shape(Q) = S+t = \shape(P) + t$.

\noindent 2. $P \nar{\mu} Q$ implies $Q= S[B']$ for a proper $B'\in \bpa$ (see rule \name{Basic$_c$}
in Table~\ref{table:proc-behaviour}). Hence $\shape(Q) = S = \shape(P)$.

\par\bigskip\noindent
{\em Comp:} $ P = P_1 \unionc{a,X} P_2$.

\noindent 1. $P \wnar{t} Q$ implies $P_1 \wnar{t} Q_1$, $P_2 \wnar{t} Q_2$, and $Q= Q_1
\unionc{a, Y} Q_2$ where $Y = X + (t \cdot \velocity{P})$.
By induction hypothesis,  $\shape(Q) = \shape(Q_1) \union{Y}
\shape(Q_2) = (\shape(P_1)+t) \union{Y} (\shape(P_2)  + t) = \shape(P) + t$.

\noindent 2. $P \nar{\mu} Q$ implies either $P_1 \nar{\mu} Q_1$ and $Q= Q_1
\unionc{a,X} P_2$ or $P_2
\nar{\mu} Q_2$ and $Q= P_1 \unionc{a,X} Q_2$. Let us consider the former case
(the latter one can be proved similarly). By induction hypothesis it is $\shape(Q) =
\shape(Q_1) \union{X} \shape(P_2) = \shape(P_1) \union{X} \shape(P_2)  = \shape(P)$.

\par\medskip\noindent
Now we prove that $P$ well-formed and $P \wnar{t} Q$ imply $Q$ well-formed. Again, we proceed by
induction on $P$ well-formed.
\par\medskip\noindent
{\em Basic:}$ P = S[B] $. If $P \wnar{t} Q$ then $Q= (S+t)[B']$ where $B \wnar{t} B'$. In this
case, $S$ well-formed and Prop.~\ref{prop:shapes-wellformedness}
imply $\shape(Q) = S + t$ well-formed. Moreover $\reference(X, \referencepoint{P}) \subseteq
\boundary{P}$ for each site $X \subseteq \pos$ that occurs in $B$ and, hence, in $B'$. By
Lemma~\ref{lemma:shape-properties1}-4,  $\reference(X, \referencepoint{Q}) = \reference(X,
\referencepoint{P} + t \cdot \velocity{P}) = \reference(X, \referencepoint{P})  + t \cdot
\velocity{P} \subseteq \boundary{P} + t\cdot \velocity{P} = \boundary{Q}$. Thus,  $Q$
is well-formed.

\par\medskip\noindent
{\em Comp:} $ P = P_1 \unionc{a,X} P_2$ with $P_i$ well-formed for $i=1,2$ and $X
\subseteq \boundary{P_1} \cap \boundary{P_2}$.
$P \wnar{t} Q$ implies $P_i \wnar{t} Q_i$, for $i=1,2$, and $Q= Q_1 \unionc{a, Y} Q_2$ with
$Y = X + (t \cdot \velocity{P})$. By induction hypothesis, $Q_1$ and $Q_2$ are well-formed.
Moreover, $X \subseteq \boundary{P_1} \cap \boundary{P_2}$ implies $Y = X + t \cdot
\velocity{P} \subseteq (\boundary{P_1} \cap \boundary{P_2}) + t \cdot \velocity{P} =
(\boundary{P_1} + t \cdot \velocity{P}) \cap (\boundary{P_2} + t \cdot \velocity{P}) =
\boundary{Q_1} \cap \boundary{Q_2}$.

\par\medskip\noindent
It remains to prove, by induction on $P$, that $P$ well-formed and $P \nar{\mu} Q$ imply $Q$ well-formed. 

\par\medskip\noindent
{\em Basic:} $P = S[B]$.
If $P \nar{\mu} Q$ then $B \nar{\mu} B'$ for a proper $B' \in \proc$ and $Q= S[B']$. $P$ well-formed
implies $S$ well-formed. Moreover, for each $X$ that occurs in $B'$ (and hence in $B$),
$\reference(X, \referencepoint{P}) \subseteq \boundary{P} = \boundary{S} = \boundary{Q}$. So, $Q$
is well-formed.
\par\medskip\noindent
{\em Comp:} $ P = P_1 \unionc{a,X} P_2$ with $P_i$ well-formed for $i=1,2$ and $X
\subseteq \boundary{P_1} \cap \boundary{P_2}$.
$P \nar{\mu} Q$ implies either $P_1 \nar{\mu} Q_1$ and $Q= Q_1 \unionc{a,X} P_2$ or $P_2
\nar{\mu} Q_2$ and $Q= P_1 \unionc{a,X} Q_2$. We only prove the former case (the latter one is
similar). By induction hypothesis $Q_1$ is well-formed; moreover, $\shape(Q_1) = \shape(P_1)$ and,
hence, $\boundary{Q_1} = \boundary{P_1}$.
Thus, $Q_1$ and $P_2$ are well-formed with $X \subseteq \boundary{P_1} \cap
\boundary{P_2} = \boundary{Q_1} \cap \boundary{P_2}$, i.e.\ $Q$ is well-formed.
\end{proof}
\par\medskip\noindent
{\bf Proposition~\ref{prop:splitwf}}
\em Let $P \in \proc$ well-formed. If $\langle a, X \rangle \in \bonds(P)$ there is a well-formed 3D
process $Q \unionc{a, X} R \equiv_P P$.\rm

\smallskip\noindent
\begin{proof}  By induction on $P$ well-formed.

\par\medskip\noindent
{\em Basic:} $ P = S[B] $. This case is not possible since $\bonds(P) = \emptyset$.
\par\medskip\noindent
{\em Comp:} $ P = P_1 \unionc{b,Y} P_2$ with $P_i$ well-formed for $i=1,2$, $Y
\subseteq \boundary{P_1} \cap \boundary{P_2}$ and $\velocity{P_1} = \velocity{P_2} = \velocity{P}
= \vect{v}$.
If $\langle a, X \rangle = \langle b, Y \rangle$ it suffices to choose $Q=P_1$ and $R=P_2$. Assume
$\langle a, X \rangle \in \bonds(P_1)$ (the case in which $\langle a, X \rangle \in \bonds(P_2)$
is symmetric).  By induction hypothesis, $Q_1 \unionc{a,X} R_1 \equiv_P P_1$ is
well-formed and $P = P_1 \unionc{b,Y} P_2 \equiv_P (Q_1 \unionc{a, X} R_1) \unionc{b,Y} P_2 $.
We have  two possible subcases:

\noindent 1. $Y \subseteq \boundary{R_1}$ and $P \equiv_P Q  \unionc{a, X} R$ where $Q = Q_1$ and $R
= R_1 \unionc{b, Y} P_2$, 

\noindent 2. $Y \subseteq \boundary{Q_1}$ and $P \equiv_P Q \unionc{a, X} R$ where $Q = R_1$ and $R
= Q_1 \unionc{b, Y} P_2$.
We only consider the former case (the latter one is similar) and prove that $Q \unionc{a,X} R$ is 
well-formed. To this aim we have to show that: ({\it i}) $R = R_1 \unionc{b, Y} P_2$ is
well-formed (note that $P_1 \equiv_P Q_1 \unionc{a,X} R_1$ well-formed and $Q = Q_1$ implies $Q_1$
well-formed), ({\it ii}) $\velocity{Q} = \velocity{R}$ (this follows easily because $P$ well-formed
implies $\velocity{Q_1}$ (i.e.\ $\velocity{Q}$)$= \velocity{R_1}  = \velocity{P_2} =\velocity{R}$)
and ({\it iii}) $X \subseteq \boundary{Q} \cap \boundary{R}$. Below we prove first ({\it i}) and
then ({\it iii})

\par\smallskip\noindent $P_1 \equiv_P Q_1 \unionc{a,X} R_1$ and $P = P_1 \unionc{b,Y} P_2$
well-formed imply, resp., $R_1$ and $P_2$ well-formed; moreover, $\velocity{R_1} = \velocity{P_1} =
\vect{v} = \velocity{P_2}$ (because $P_1 \equiv_P Q_1 \unionc{a,X} R_1$ is well-formed). To show
that $R = R_1 \unionc{b, Y} P_2$ is well-formed,  it remains to prove that $Y \subseteq
\boundary{R_1} \cap \boundary{P_2}$. We proceed by contradiction. Assume $\vect{y} \in Y$ such that
$\vect{y} \notin \boundary{R_1} \cap \boundary{P_2}$. Then, $Y \subseteq \boundary{R_1}$ implies
$\vect{y} \notin \boundary{P_2}$ and, hence, $\vect{y} \notin \boundary{P_1} \cap \boundary{P_2}$.
But, this is impossible because $Y \subseteq \boundary{P_1} \cap \boundary{P_2}$.

\par\smallskip\noindent
We prove that $X \subseteq \boundary{Q} \cap \boundary{R} = \boundary{Q_1} \cap \boundary{R}$ by
contradiction. Recall that $Q_1 \unionc{a,X} R_1$  well-formed implies $X \subseteq \boundary{Q_1}
\cap \boundary{R_1}$. Assume that there is $\vect{x} \in X \subseteq \boundary{Q_1} \cap
\boundary{R_1}$ such that $\vect{x} \notin \boundary{Q_1} \cap \boundary{R}$. Then $\vect{x} \notin
\boundary{R} = (\boundary{R_1} \cup \boundary{P_2}) \backslash \{\vect{z} \,|\, \vect{z} \mbox{ is
interior of } \points{R_1} \cup \points{P_2}\}$. Moreover, $\vect{x} \in X \subseteq \boundary{R_1}
\subseteq \boundary{R_1} \cap \boundary{P_2}$ and $\vect{x} \notin \boundary{R}$ implies that
$\vect{x}$ is an interior point of $\points{R_1} \cup  \points{P_2}$, i.e. $\vect{x}$ is an interior
point of $\points{P_2}$; this is  because $\vect{x} \in \boundary{R_1}$ implies that $\vect{x}$
cannot be an interior point of $\points{R_1}$. Finally, $\vect{x} \in X$ implies that $\vect{x}$ is
also an interior  point of $P_1 \equiv_P Q_1 \unionc{a, X} R_1$. Thus, $P_1$ and $P_2$
interpenetrate each other. But this is not possible since $P \equiv_P P_1 \unionc{b, Y} P_2$ is 
well-formed.
\end{proof}

\par\bigskip\noindent
{\bf Proposition~\ref{prop:closure-splittings}}{\it
Let $P, Q\in \proc$  and $\mu \in \omega(\channels) \cup \rho(\channels)$. Then:

\par\medskip\noindent 1. $P \dnar{\mu} Q$ implies $\shape(Q) = \shape(P)$;

\par\smallskip\noindent 2.  $P$ well-formed and $P\dnar{\mu} Q$ implies $Q$ well-formed.}

\par\smallskip\noindent \begin{proof}
We only prove the statement for $\mu = \rho(a,X)$; if $\mu = \omega(a,X)$ the statement can be
proved similarly.

\par\medskip\noindent
{\em Basic:} $ P = S[B] $. This case is not possible since $P
\not{\dnar{\rho(a,X)}}$.
\par\medskip\noindent
{\em Comp:} $ P = P_1 \unionc{b,Y} P_2$ with $P_i$ well-formed for $i=1,2$ and $X \subseteq
\boundary{P_1} \cap \boundary{P_2}$. We distinguish two possible subcases:

\noindent ({\it i}) $\langle b, Y \rangle = \langle a, X \rangle$, $P_1 \nar{\rho(\alpha, X_1)}
Q_1$, $P_2 \nar{\rho(\overline{\alpha}, X_2)} Q_2$ with $\alpha \in \{a, \overline{a}\}$ and $X =
X_1 \cap X_2$, and $Q= Q_1 \unionc{a,X} Q_2$.

\par\medskip\noindent 1. By Prop.~\ref{prop:closure-3d}-2, $\shape(Q_i) = \shape(P_i)$ for $i=1,2$
and, hence, $\shape(Q) = \shape(Q_1) \union{X} \shape(Q_2) = \shape(P_1) \union{X} \shape(P_2)  =
\shape(P)$.

\par\smallskip\noindent 2. By induction hypothesis, $P_i$ well-formed implies $Q_i$ well-formed, for
$i=1,2$. Moreover $X 
\subseteq \boundary{P_1} \cap \boundary{P_2}$ and, again, $\shape(Q_i) = \shape(P_i)$ for $i=1,2$
imply $X \subseteq \boundary{Q_1} \cap \boundary{Q_2}$. So the 3D process $Q$ is well-formed.

\par\medskip\noindent ({\it ii}) Either $P_1 \dnar{\rho(a,X)} Q_1$ and $Q= Q_1
\unionc{a,X} P_2$ or $P_2 \dnar{\rho(a,X)} Q_2$ and
$Q= P_1 \unionc{a,X} Q_2$.

\par\medskip\noindent 1. By induction hypothesis it is $\shape(Q_i) = \shape(P_i)$ for $i=1,2$. As
in the previous case, we can prove that $\shape(Q) = \shape(P)$.

\par\smallskip\noindent 2. By induction hypothesis we have that $Q_i$ and  $P_j$, for $\{i,j\} \in
\{1,2\}$, are well-formed. Moreover, by Item 1, $\shape(Q_i)= \shape(P_i)$, and hence $X \subseteq
\boundary{P_1} \cap \boundary{P_2} = \boundary{Q_i} \cap \boundary{P_j}$. Also in this case we can
conclude that $Q$ is well-formed.
\end{proof}

\section*{Appendix C: Proofs of Section~\ref{sec:networks}}
\par\medskip\noindent
{\bf Proposition~\ref{prop:split-closure}}
\em Let $P \in \proc$ well-formed and $C \subseteq \channels$. Then $\Split(P,C)$ is
a well-formed network of 3D processes \rm.

\medskip

\begin {proof}
By induction on the number of channels in $\bonds(P) \cap C$.
\par\medskip\noindent
$\bonds(P) \cap C = \emptyset$. In such a case $\Split(P,C) = P$  is  well-formed.
\par\medskip\noindent
$\langle a, X \rangle \in \bonds(P) \cap C \neq \emptyset$. By
Prop.~\ref{prop:splitwf}, there are $P_1, P_2$ well-formed such that $P
\equiv_P P_1 \unionc{a, X} P_2$ and $\Split(P,C) = \Split(P_1, C) \,\|\, \Split(P_2, C) = N_1 \,\|\,
N_2$. By induction hypothesis $N_i$ is a well-formed 3D network for $i=1,2$. Now, in
order to prove that $\Split(P, C)$ is  well-formed, it still remains to prove that if $Q_1$
and $Q_2$ are two 3D processes in the network $N_1$ and $N_2$,  respectively, then $Q_1$ and $Q_2$
do not interpenetrate.
Assume, towards a contradiction, that there are $Q_1$ and $Q_2$ in $N_1$ and $N_2$, respectively,
such that $Q_1$ and $Q_2$ interpenetrate, i.e.\ such that there is (at least) a point $\vect{x}$
that is interior of both $Q_1$ and $Q_2$.

For each $i=1,2$, it holds that either $\bonds(P_i) \cap C = \emptyset$,  $N_i =
P_i$ and $Q_i = P_i$ or $\bonds(P_i) \cap C = \{\unionc{a_1, X_1}, \unionc{a_2, X_2},.. ,
\unionc{a_{ni}, X_{ni}}\}$, $P_i \equiv_P P_i^1 \unionc{a_1, X_1} (P_i^2 \unionc{a_2, X_2} ..
(P_i^{ni} \unionc{a_{ni}, X_{ni}} P_i^{ni + 1}))$ (this follows by iterating\\ 
Prop.~\ref{prop:splitwf}), $N_i = P_i^1 \,\|\, (P_i^2 \,\|\, \cdots (P_i^{ni} \,\|\, P_i^{ni+1}))$
and $Q_i = P_i^j$ for some $j = 1,2, \cdots, ni, ni+1$. Thus, in both cases the set of interior
points of $Q_i$ is a subset of interior point of $P_i$. As a consequence, if
$Q_1$ and $Q_2$ interpenetrate so do $P_1$ and $P_2$. But this is not possible because
$P$ is well-formed.
\end {proof}
To prove Prop.~\ref{prop:net-closure} we need the following preliminary result.
\begin{lemma}\label{lemma:net-time}
Let $N, M \in \nets$, $P, Q \in \proc$ and $t \in \timedomain$ such that $N \nar{t} M$ and $P
\nar{t} Q$. The network $N$ contains $P$ iff $M$ contains $Q$.
\end{lemma}

\begin{proof}
By induction on $N \in \nets$.
\par\smallskip\noindent
$N = \NIL$.  $N \nar{t} \NIL$ and $\NIL$ does not contain any 3D
process.
\par\smallskip\noindent
 $N = P$ with $P\in \proc$. By our operational semantics, $N \nar{t} M$ iff $M = Q$ and
$P \nar{t} Q$.
\par\smallskip\noindent 
$N = N_1 \,\| \, N_2$.  $N \nar{t} M_1 \,\| \, M_2 = M$ iff  $N_i \nar{t} M_i$ for $i=1,2$.
$P$ is contained in $N$ iff $P$ is contained either in $N_1$ or in $N_2$ iff (by induction
hypothesis) and $Q$ is contained either in $M_1$ or in $M_2$, i.e.\ iff $Q$ is contained in $M$.
\end{proof}

\par\medskip\noindent {\bf Proposition~\ref{prop:net-closure}}
\em Let $t \in \timedomain$, $P\in \proc$, $N\in \nets$, with $P$ and $N$ well-formed.

\par\smallskip\noindent 1. $P \nar{\omega} N$ and $P \nar{\rho} N$ implies $N \in \nets$
well-formed.

\par\smallskip\noindent 2. $N \nar{t} M$ implies $M$ well-formed;

\par\medskip\noindent 3. $N \dnar{t} M$ implies $M$ well-formed.\rm

\par\medskip\noindent \begin{proof}
Item 1 can be proved by (iterative applications of) Prop.~\ref{prop:closure-splittings}-2 and
Prop.~\ref{prop:split-closure}. Item 3 follows directly from Items 1 and 2 (see
Definition~\ref{def:nets-behaviour}). In what follows, we prove Item 2 by induction on $N\in \nets$.
\par\medskip\noindent
 $N = \NIL$. In such a case $N \nar{t} M$ implies $M=\NIL$ that is  well-formed.
\par\medskip\noindent
$N = P$ with $P\in \proc$. The statement is a trivial consequence of
Prop.~\ref{prop:closure-3d}-3.
\par\medskip\noindent 
$N = N_1 \,\| \, N_2$. If $N \nar{t} M$  then $N_i \nar{t} M_i$, for $i=1,2$, and $M  = M_1 \,\| \,
M_2$. By induction hypothesis, $M_1$ and $M_2$ are well-formed. So, it remains to prove that if
$Q_1$ and $Q_2$ are 3D processes that compose the networks $M_1$ and $M_2$, resp., then $Q_1$ and
$Q_2$ do not interpenetrate. We proceed by contradiction. Assume that there are $Q_1$ in $M_1$ and
$Q_2$ in $M_2$ that interpenetrate each other. By Lemma~\ref{lemma:net-time},  $M_i$ contains $Q_i$
iff $N_i$ contains $P_i$ for some $P_i \in \proc$ such  that $P_i \nar{t} Q_i$. Moreover, by
Prop.~\ref{prop:closure-3d}-1, $\shape(Q_i) = \shape(P_i) + t$. Thus, if $Q_1$ and $Q_2$
interpenetrate the same do $P_1$ and $P_2$. But this not possible since $N$ is well-formed.
\end{proof}

\begin{lemma}\label{lemma:reduction-steps}
Let $N, M \in \nets$ such that $N$ is well-formed and $N \nar{\langle P, Q, X \rangle} M$ for some 
$\langle P,Q, X \rangle \in \colliding(N)$. If $P'$ and $Q'$ are two processes composing $M$ with
$\shape(P') = \shape(P)$ and $\shape(Q') = \shape(Q)$, then $\langle P',Q', Y \rangle \notin
\colliding(M)$ for any $Y \subseteq \pos$.
\end{lemma}

\begin{proof}
By rules in Table~\ref{table:reduce},  $N \enar{\langle P, Q, X \rangle} M$ implies $N \enar{\langle
P, Q, X \rangle} M$ implies $M \equiv (P \lsbrace \vect{v}_p\rsbrace \,\|\, Q\lsbrace \vect{v}_q
\rsbrace)
\,\|\, N'$ where the pair of velocity $(\vect{v}_p, \vect{v}_q) = P \collision{X}{e} Q$.
Moreover,  Equations (1) and (2) in Def.~\ref{def:collisions} ensure that $P \lsbrace
\vect{v}_p\rsbrace$ and $Q\lsbrace \vect{v}_q \rsbrace$ (since $N$ is well-formed, these are the
only processes in $M$ with the same shape of $P$ and $Q$, resp.) can not collide anymore.
Now assume $N \inar{\langle P, Q, X \rangle} M$. By Def.~\ref{def:reductionrules}, there are
compatible $\langle \alpha, X_p\rangle, \langle \overline{\alpha}, X_q \rangle \in \channels$ 
such that $P \nar{\langle \alpha, X_p\rangle} P'$, $Q \nar{\langle
\overline{\alpha}, X_q \rangle} Q'$ and $M \equiv  ((P' \unionc{a,X} Q') \lsbrace \vect{v} \rsbrace)
\,\|\, N'$ with $\alpha \in \{a,\overline{a}\}$, $X = X_p \cap X_q \neq \emptyset$ and $\vect{v} =
P \collision{X}{i} Q$. The statement follows because there  are no $P'$ and $Q'$ in $M$ such
that $\shape(P') = \shape(P)$  and $\shape(Q') = \shape(Q)$.

\end{proof}

\par\bigskip\noindent{\bf Proposition~\ref{prop:reduction-closure}}
\em Let $N \in \nets$, $P,Q \in \proc$ and $X$ a not-emptyset subset of $\pos$. Then $N$ well-formed
and $N \nar{\langle P, Q, X \rangle} M$ implies $M$ well-formed. \rm

\par\medskip\noindent \begin{proof}
If  $N \enar{\langle P, Q, X \rangle} M$  then $N \equiv (P \,\|\, Q) \,\|\, N'$, $M \equiv  (P
\lsbrace \vect{v}_p\rsbrace \,\|\, Q\lsbrace \vect{v}_q \rsbrace) \,\|\, N'$ and $(\vect{v}_p,
\vect{v}_q) = P \collision{X}{e} Q$. Moreover, $N$ well-formed implies that $P$, $Q$ and $N'$ are
well-formed and processes $P$ and $Q$ do not interpenetrate each other and with any other 3D process
composing $N'$ (see Def.~\ref{def:networks}).
Finally:  $P \lsbrace \vect{v}_p\rsbrace$ and $Q\lsbrace \vect{v}_q
\rsbrace$ are well-formed 3D processes with $\shape(P \lsbrace \vect{v}_p\rsbrace) = \shape(P)$ and
$\shape(Q\lsbrace \vect{v}_q \rsbrace) = \shape(Q)$ (see Lemma~\ref{lemma:shape-properties2}).
Thus, $P\lsbrace \vect{v}_p \rsbrace$ and $Q\lsbrace \vect{v}_q \rsbrace$ do not interpenetrate
each other and with any other 3D process in $N'$ and we can conclude that $M$ is well-formed.

\par\medskip\noindent Now assume $N \inar{\langle P, Q, X \rangle} M$. Again, $N \equiv (P \,\|\, Q)
\,\|\, N'$ and, by rule~\name{inel},  there
are $\langle \alpha, X_p\rangle, \langle \overline{\alpha}, X_q \rangle \in \channels$ compatible
such that $\alpha \in \{a,\overline{a}\}$, $P \nar{\langle \alpha, X_p\rangle} P'$, $Q \nar{\langle
\overline{\alpha}, X_q \rangle} Q'$ and $M \equiv  ((P' \unionc{a,X} Q') \lsbrace \vect{v} \rsbrace)
\,\|\, N'$ where $X = X_p \cap X_q \neq \emptyset$ and $\vect{v} = P \collision{X}{i} Q $. By
Prop.~\ref{prop:closure-3d}-3, $P$ and $Q$ well-formed implies  $P'$ and $Q'$ well-formed. Moreover:

\par\smallskip\noindent (1) $P \nar{\langle \alpha, X_p\rangle} P'$ and $Q \nar{\langle
\overline{\alpha}, X_q \rangle} Q'$ imply $X_p \subseteq \boundary{P}$ and $X_q \subseteq
\boundary{Q}$, resp. Thus, $X = X_p \cap X_q \subseteq \boundary{P} \cap \boundary{Q} =
\boundary{P'} \cap \boundary{Q'}$ (since Prop.~\ref{prop:closure-3d}-2 implies
$\shape(P')= \shape(P)$ and $\shape(Q')= \shape(Q)$), $P' \unionc{a,X} Q'$ and $(P' \unionc{a,X} Q')
\lsbrace \vect{v} \rsbrace$ are well-formed 3D processes.

\par\smallskip\noindent (2) $\shape((P' \unionc{a,X} Q') \lsbrace \vect{v} \rsbrace) =
\shape(P' \unionc{a,X} Q') = \points{P'} \cup \points{Q'} = \\
\points{P} \cup \points{Q}$. As a consequence, $(P' \unionc{a,X} Q') \lsbrace \vect{v} \rsbrace$ can
not interpenetrate any 3D process that compose the network $N'$ (otherwise either $P$ or $Q$ must
do the same).

\par\medskip\noindent Aso in this case we can conclude that $M$ is well-formed.
\end{proof}

\end{document}